\documentclass[11pt]{article}
\pdfoutput=1
\usepackage{etoolbox,xstring}
\providebool{csub}
\setbool{csub}{false}

\usepackage[utf8]{inputenc}
\usepackage{amsmath, amssymb, amsfonts, bm, cite}
\renewcommand{\mathbf}[1]{\boldsymbol{#1}}

\usepackage{nicefrac}

\ifbool{csub}{%
  \usepackage{tgtermes}
}{%
  \usepackage{libertine}
}
\usepackage[T1]{fontenc}

\usepackage{booktabs,siunitx}

\usepackage{amsthm}
\usepackage[colorlinks=true,citecolor=blue,bookmarksnumbered=true,hyperfootnotes=false]{hyperref}

\newtheorem{theorem}{Theorem}[section]
\newtheorem{observation}[theorem]{Observation}

\newtheorem{lemma}[theorem]{Lemma}
\newtheorem{corollary}[theorem]{Corollary}
\theoremstyle{definition}
\newtheorem{definition}{Definition}[section]
\newtheorem{remark}{Remark}[section]

\newcommand{\abs}[1]{\ensuremath{\left|#1\right|}}
\newcommand{\norm}[2][]{\ensuremath{\Vert #2 \Vert_{#1}}}

\newcommand{\diff}[2]{\frac{\text{d}#1}{\text{d}#2}}
\newcommand{\pdiff}[2]{\frac{\partial #1}{\partial #2}}

\newcommand{\E}[2][]{\ensuremath{\mathbb{E}_{#1}\insq{#2}}}

\newcommand{\ina}[1]{\left<#1\right>}
\newcommand{\inb}[1]{\left\{#1\right\}}
\newcommand{\inp}[1]{\left(#1\right)}
\newcommand{\insq}[1]{\left[#1\right]}

\makeatletter
\newcommand*{\defeq}{\mathrel{\rlap{%
                     \raisebox{0.3ex}{$\m@th\cdot$}}%
                     \raisebox{-0.3ex}{$\m@th\cdot$}}%
                    =}
\makeatother

\newcommand{\nfrac}[3][]{\nicefrac[#1]{#2}{#3}}

\newcommand{\Z}[0]{\ensuremath{\mathbb{Z}}}
\newcommand{\R}[0]{\ensuremath{\mathbb{R}}}

\newcommand{\poly}[1]{\ensuremath{\mathop{\mathrm{poly}}\inp{#1}}}

\renewcommand{\vec}[1]{\bm{#1}}

\newcommand{\Ts}[1]{\ensuremath{T_{SAW}\inp{#1}}}

\newcommand{\G}[0]{\mathcal{G}}
\newcommand{\br}[1]{\mathrm{br}\inp{#1}}

\newcommand{\family}[1]{{\mathcal{#1}}}

\usepackage[paperwidth=8.5in,paperheight=11.0in,margin=1.0in]{geometry}

\begin{document}
\title{Spatial mixing and approximation algorithms\\for graphs with
  bounded connective constant}
\newcommand{\sincgrant}{NSF grant CCF-1016896} 
  \author{Alistair Sinclair \thanks{Computer Science Division, UC
      Berkeley. Email: \texttt{sinclair@cs.berkeley.edu}. Supported in
      part by \sincgrant~and by the Simons Institute for the Theory of
      Computing.}
    \and Piyush Srivastava\thanks{Computer Science Division, UC
      Berkeley. Email: \texttt{piyushsriva@gmail.com}. Supported in part
      by \sincgrant.}
    \and Yitong Yin\thanks{State Key Laboratory for Novel Software Technology, Nanjing University, China. Email:
      \texttt{yinyt@nju.edu.cn}. Supported by NSFC grants 61272081 and
      61021062. Part of this work was done while this author was
      visiting UC Berkeley.}}  

\date{}

\maketitle
\begin{abstract}
  The hard core model in statistical physics is a probability
  distribution on independent sets in a graph in which the weight of
  any independent set $I$ is  proportional to $\lambda^{|I|}$, where
  $\lambda > 0$ is the \emph{vertex activity}.  We show that there is an intimate connection between the \emph{connective constant} of a
  graph and the phenomenon of strong spatial mixing (decay of
  correlations) for the hard core model; specifically, we prove that the hard core model with vertex
  activity $\lambda < \lambda_c(\Delta + 1)$ exhibits strong spatial
  mixing on any graph of connective constant $\Delta$, irrespective of
  its maximum degree, and hence derive an FPTAS for the partition
  function of the hard core model on such graphs.  Here $\lambda_c(d)
  \defeq \frac{d^d}{(d-1)^{d+1}}$ is the critical activity for the
  uniqueness of the Gibbs measure of the hard core model on the
  infinite $d$-ary tree.  As an application, we show that the
  partition function can be efficiently approximated with high
  probability on graphs drawn from the random graph model
  $\G\inp{n,d/n}$ for all $\lambda < e/d$, even though the maximum degree
  of such graphs is unbounded with high probability.

  We also improve upon Weitz's bounds for strong spatial mixing on
  bounded degree graphs~\cite{Weitz06CountUptoThreshold} by providing
  a computationally simple method which uses
  known estimates of the connective constant of a lattice to obtain
  bounds on the vertex activities $\lambda$ for which the hard
  core model on the lattice exhibits strong spatial mixing.  Using
  this framework, we improve upon these bounds for several lattices
  including the Cartesian lattice in dimensions $3$ and higher.

  Our techniques also allow us to relate the threshold for the
  uniqueness of the Gibbs measure on a general tree to its
  \emph{branching factor}~\cite{lyons_ising_1989}.
  
\end{abstract}
\thispagestyle{empty}
\newpage
\setcounter{page}{1}
\section{Introduction}
\label{sec:introduction}
\subsection{Background}
\label{sec:background}
In spin systems on graphs, the property of {\it spatial mixing\/} (i.e., 
decay of correlations with distance) plays a central role.  In statistical
physics spatial mixing guarantees a unique Gibbs measure, and thus
a single phase for the underlying physical model.  In computer science
spatial mixing implies the existence of efficient algorithms for
approximating key combinatorial quantities.  Much attention has therefore
been focused on identifying ranges of parameters for which spatial mixing
holds; interestingly, this is a case where the computer science perspective,
motivated by algorithmic applications, has led to new insights into the
behavior of physical models.

In this paper we contribute to this line of work, focusing on the {\it hard core\/}
(or {\it weighted independent set\/}) model, which is one of the most widely
studied classical examples 
(though our techniques actually apply more
 widely to other two-spin systems, such as the anti-ferromagnetic Ising model).
The
configurations of the hard core model are
the independent sets of a graph $G=(V,E)$, each of which has a weight
$w(I)=\lambda^{|I|}$, where $\lambda$ is a positive parameter known as
the {\it vertex activity}.  The probability of occurrence of configuration~$I$
is determined by the {\it Gibbs distribution\/}: $$
   \pi(I) = w(I)/Z.  $$
Here the normalizing factor $Z=Z(G,\lambda) \defeq \sum_I w(I)$ is called
the {\it partition function}.  As a natural
generalization of counting, computing $Z$ is a central problem in
statistical physics and combinatorics, and is known to be \#P-hard in
most interesting cases.  Considerable interest has been devoted to the
problem of approximating $Z$, which is equivalent to
sampling from the Gibbs distribution~\cite{jervalvaz86}. 

In the special case where $G$ is the infinite $d$-ary tree, it is
well known that the hard core model exhibits a phase transition:
there exists a
critical activity $\lambda_c(d) = \frac{d^d}{(d-1)^{d+1}}$ such
that point-to-set
correlations in the Gibbs distribution\footnote{When $G$ is an
  infinite graph the partition function is not well defined.  However,
  the Gibbs distribution is still well-defined on any finite subset
  of~$G$, and can be extended in a natural way to a measure on the
  whole of~$G$.  On the infinite $d$-ary tree, this extension is
  uniquely determined only when $\lambda \leq \lambda_c(d)$.}
decay to zero exponentially with distance when $\lambda < \lambda_c(d)$ (this
decay is referred to as \emph{weak spatial mixing}), and remain
bounded away from zero when $\lambda > \lambda_c(d)$.  This phase transition
has been at the center of dramatic recent results relating
phase transitions to the computational
complexity of partition functions~\cite{Weitz06CountUptoThreshold,Sly2010CompTransition}.

A stronger notion of decay of correlations is that of \emph{strong
spatial mixing}.  Here, the exponential decay of point-to-set
correlations is required to hold even in the presence of arbitrary
boundary conditions; i.e.,  correlations should still decay even 
when the configurations of some vertices (possibly close to those
whose correlations are being measured) are fixed arbitrarily.  
Algorithmically, strong spatial mixing guarantees the existence of
an efficient approximation algorithm (an FPTAS) for the partition 
function~$Z$ and other important quantities associated with the model.
In a seminal
paper~\cite{Weitz06CountUptoThreshold}, Weitz showed that for the hard
core model, strong spatial mixing on the $d$-ary tree implies strong
spatial mixing on \emph{any} graph of degree at most $d+1$.  Weitz further showed
that weak spatial mixing is equivalent to strong spatial mixing
for the hard-core model on the $d$-ary tree, thus establishing that a
graph of maximum degree $d+1$ always exhibits strong spatial mixing,
and hence admits an FPTAS for the partition function,
for all $\lambda < \lambda_c(d)$.

At the time of publication of Weitz's paper, his bound for graphs of
maximum degree $d+1$ in terms of the
critical activity of the $d$-ary tree improved upon the best known bounds
for spatial mixing even for %
\ifbool{csub}{%
  widely studied special classes of graphs such as Cartesian
  lattices. 
}{%
  such special classes of graphs as
  Cartesian lattices, which are the most widely studied in statistical physics.
}
This remained the state-of-the-art until the recent work of
Restrepo, Shin, Tetali, Vigoda and Yang~\cite{restrepo11:_improv_mixin_condit_grid_count}, who
improved upon Weitz's bounds in the special case of the 2-dimensional Cartesian
lattice.  Although in principle the methods of Restrepo {\it et al.}\
can be applied to any fixed lattice, such an application
requires a numerical search over a high-dimensional parameter
space, and hence is computationally very intensive.  Further, in contrast to Weitz's bound, which depends
solely upon the maximum degree of the graph, the bounds of Restrepo
{\it et al.} (as well as those in the subsequent paper of Vera, Vigoda and Yang~\cite{vera13:_improv_bound_phase_trans_hard})  do not seem
to depend upon any easily identifiable characteristic of the graph.
Also, their methods are tailored to \emph{fixed}
lattices rather than to general classes of
graphs (such as graphs drawn from the random graph family $\G(n, d/n)$).   

\subsection{Contributions}
\label{sec:contributions}
In this work we improve upon Weitz's result for the hard core model in
two ways.  First, we relax the bounded degree restriction in Weitz's
result.  Second, in the case of bounded degree graphs we improve
Weitz's bounds by taking into account more information about the
structure of the graph.  We achieve these goals by relating spatial
mixing on a graph to its \emph{connective constant}, a natural measure of the
``effective degree'' in a sense that we describe below.  Since in many interesting families
of graphs the connective constant is significantly smaller than the
maximum degree, this will allow us to obtain tighter bounds.

For an infinite graph $G=(V,E)$, the connective constant $\Delta(G)$
is defined as 
\ifbool{csub}{%
  $\sup_{v\in V}\limsup_{\ell\rightarrow\infty}N(v,\ell)^{1/\ell},$
}{
  $$\sup_{v\in V}\limsup_{\ell\rightarrow\infty}N(v,\ell)^{1/\ell},$$%
} %
where $N(v,\ell)$ is the
number of self-avoiding walks of length~$\ell$ in~$G$ starting at~$v$.
This definition extends naturally to families of finite graphs (see
Section~\ref{sec:connective-constant}).  The set of self-avoiding walks originating at a vertex $v$ can naturally
be viewed as a tree rooted at $v$, known as the
\emph{self-avoiding walk (SAW) tree}.  The connective constant can
then be viewed as the \emph{average arity} of this tree.  Thus, for
example, the connective constant of a graph of maximum degree $d+1$ is
at most $d$, though it can be much smaller. The
connective constant is a well-studied quantity, especially for
standard graph families such as lattices, and rigorous bounds on its
value are known in many cases (see, for example,
\cite{alm_upper_2005,alm_upper_1993,poenitz00:_improv_z,jensen_enumeration_2004,madras96:_self_avoid_walk} and the
recent breakthrough in \cite{duminil-copin_connective_2012}).

Our interest in the connective constant comes from Weitz's
construction in \cite{Weitz06CountUptoThreshold} establishing that the
decay of correlations on a graph is always at least as rapid as on the
corresponding SAW tree.  Intuitively, the decay of correlations on this
tree should in turn be related to the rate of growth with $\ell$ of the number of vertices at
distance $\ell$ from the root (which is exactly $N(v,\ell)$). So far,
in the case of general graphs, this intuition has been captured only
by crudely bounding the growth of the number of vertices as $d^\ell$, where $d
+ 1$ is the maximum degree of the graph.  Our results show that by instead
using the connective constant, one can obtain tighter relations between
the rate of decay of correlations and the growth of the number of
descendants.   We say a few words about our proof techniques at the end of
this section.

Our first result is an analog of Weitz's bound, with the maximum degree
replaced by the connective constant and a slightly stronger condition
on the vertex activity $\lambda$.

\begin{theorem}
  \label{thm:main-1}
  Let $\lambda>0$ and $\Delta$ be such that $\lambda <
  \lambda_c(\Delta + 1)$.  %
Then:
  \begin{enumerate}
  \item The hard core model with vertex activity~$\lambda$
    exhibits strong spatial mixing on any family of (finite or
    infinite) graphs with connective constant at most~$\Delta$.
  \item There is an FPTAS for the
    partition function of the hard core model with vertex activity~$\lambda$
    on any family of finite graphs with connective constant at most~$\Delta$.
  \end{enumerate}
\end{theorem}
\begin{remark}
  The condition on $\lambda$ in the theorem is satisfied whenever
  $\lambda < \frac{e}{\Delta}$.  This is asymptotically optimal for
  large~$\Delta$ since $\lambda_c(\Delta)\sim e/\Delta$ as
  $\Delta\to\infty$. (We note that a trivial path coupling argument can
  be used to prove strong spatial mixing under the much stronger
  assumption $\lambda < 1/\Delta$.)
\end{remark}
\begin{remark}
If we could replace $\lambda_c(\Delta+1)$ in the above theorem by
$\lambda_c(\Delta)$, this would exactly parallel Weitz's bound with the
maximum degree replaced by the connective constant and would be optimal 
for all~$\Delta$.  Although we do not know if such a result holds in
general, we do obtain such a tight bound for the special case of
spherically symmetric trees; see %
\ifbool{csub}{%
  the full version~\cite{sinclair13:_spatial} for details. 
}{%
  Section~\ref{sec:more-results-trees}.
}
\end{remark}

The result in Theorem~\ref {thm:main-1} has the advantage
of being applicable without any bound on the maximum degree.
This is in contrast to various recent
results on the approximation of the hard core model partition
function~\cite{restrepo11:_improv_mixin_condit_grid_count,Weitz06CountUptoThreshold,li_approximate_2012,li_correlation_2011},
all of which require the maximum degree to be bounded
by a constant.  As an application of Theorem~\ref{thm:main-1}
we consider the random graph model $\G(n,d/n)$, which has 
constant average degree~$d(1-o(1))$ but unbounded maximum degree
$\Theta(\log n/\log\log n)$ with high probability.  We prove the
following result (see Section~\ref{sec:special-message} for a more
precise formulation of the high probability statements).
\begin{theorem}\label{thm:gndn-ssm-alg}
  Let $\epsilon > 0$, $d > 1$ and $\lambda <
  \frac{e}{d(1+\epsilon)}$ be fixed. Then:
  \begin{itemize}
  \item The hard core model with activity $\lambda$ on a graph $G$
    drawn from $\G(n,d/n)$ exhibits strong spatial mixing with high
    probability.
  \item There exists a deterministic algorithm which on input $\mu >
    0$, approximates the partition function of the hard core model
    with activity $\lambda > 0$ on a graph $G$ drawn from $\G(n, d/n)$
    within a multiplicative factor of $(1 \pm \mu)$, and which runs in
    time polynomial in $n$ and $1/\mu$ with high probability (over the
    random choice of the graph).
  \end{itemize}
\end{theorem}
%
%

Similar results for $\G(n,d/n)$ have appeared in the literature in the
context of rapid mixing of Glauber dynamics for the ferromagnetic Ising
model~\cite{mossel_exact_2013}, and also for the hard core
model~\cite{mossel_gibbs_2010,efthymiou13:_mcmc_g}.  Although the authors
of~\cite{mossel_gibbs_2010} do not supply an explicit range
of~$\lambda$ for which their rapid mixing results hold, an examination of their
proofs suggests that necessarily $\lambda<O\inp{1/d^2}$.  Similarly, the results of \cite{efthymiou13:_mcmc_g} hold when $\lambda < 1/(2d)$.  In contrast, our bound approaches the conjectured optimal value
$e/d$. Further, unlike ours, the results of~\cite{mossel_gibbs_2010,
  mossel_exact_2013,efthymiou13:_mcmc_g} are restricted to $\G(n, d/n)$ and
certain other classes of sparse graphs.

While it is asymptotically optimal, Theorem~\ref{thm:main-1} above 
is not strong enough to improve upon Weitz's result in the important
special case of small degree lattices.  To do this, we adapt our techniques
to take into account the maximum degree as well as the connective constant.
Define the function $\nu_\lambda(d)$ by
\begin{equation*}
  \nu_\lambda(d) \defeq \frac{d\tilde{x}_\lambda(d) - 1}{1 + \tilde{x}_\lambda(d)},
\end{equation*}
where $\tilde{x}_\lambda(d)$ is the unique positive solution of the
fixed point equation $dx = 1 + \lambda/(1+x)^d$. We can now state our second general theorem.
\begin{theorem}
  \label{thm:main-2}
  Let $d$ be a positive integer, and let $\lambda$ and $\Delta$ be
  such that $\nu_\lambda(d)\Delta < 1$.  Then:
  \begin{enumerate}
  \item The hard-core model with vertex activity $\lambda' \leq
    \lambda$ exhibits strong spatial mixing on any family of (finite or infinite) graphs with
    maximum degree at most $d + 1$ and connective constant at most~$\Delta$.
  \item For any $\lambda' \leq \lambda$, there is an FPTAS for the
    partition function of the hard-core model with vertex activity
    $\lambda' $ on any family of finite graphs with maximum degree at
    most $d+1$ and connective constant at most~$\Delta$.
  \end{enumerate}
\end{theorem}
\begin{remark}
  The connective constant $\Delta$ of any infinite
  graph of maximum degree $d+1$ can be at most $d$.  This, combined with the
  fact (proved later) that $\nu_\lambda(d)d < 1$ whenever
  $\lambda < \lambda_c(d)$, ensures that the bound for spatial mixing in
  Theorem~\ref {thm:main-2} is always better than that obtained
  from Weitz's result.  Thus, the above theorem extends the
  applicability of Weitz's results to a
  larger range of vertex activities.
\end{remark}

Using Theorem~\ref {thm:main-2}, we are able to improve
upon the best known spatial mixing bounds for various lattices,
including the Cartesian lattice in three and higher dimensions, 
as shown in Table~\ref {fig:1}.  The table shows, for each lattice,
the best known upper bound for the connective constant and
the strong spatial mixing (SSM) bounds we obtain using 
these values in Theorem~\ref{thm:main-2}.  In the table, a value~$\alpha$
in the ``$\Delta$" column means that SSM is shown to hold for the
appropriate lattice whenever $\lambda \leq\lambda_c(\alpha)$; the
corresponding value~$\lambda_c(\alpha)$ (rounded to three decimal places) appears in the adjacent ``$\lambda$" column.
The last pair of columns give the previously best known SSM bounds.

As is evident from the table, our general result gives improvements on 
known SSM bounds for all lattices except the 2-dimensional Cartesian 
lattice~$\mathbb{Z}^2$.  For that lattice, our bound still improves on 
Weitz's bound but is not as strong as the bound obtained 
in~\cite{restrepo11:_improv_mixin_condit_grid_count}
by more numerically intensive methods tailored to this special case.
We note, however, that any improvement in the bound on the connective
constant of the corresponding SAW tree immediately yields an improvement in our SSM
bound; %
\ifbool{csub}{%
  for an example of such an improvement, and for the details of the derivation
  of the results in the table, see the full
  version~\cite{sinclair13:_spatial} of this paper. 
}{%
  we illustrate this with a simple example in
  Appendix~\ref{sec:descr-numer-results}, where we also give details of the
  derivation of the results in the table.
}

\ifbool{csub}{%
  \begin{table*}[t]
}{%
  \begin{table}[h]
}
  \centering
  \begin{tabular}[h]{l
      S[table-format=2.6]
      c
      S[table-format=5]
      S[table-format=2.3]
      S[table-format=2.6]
      S[table-format=2.0]
      S[table-format=0.4]
    }
    \toprule
    &&&&\multicolumn{2}{c}{Our SSM bound}&\multicolumn{2}{c}{Previous
      best SSM bound}\\
    \cmidrule(r){5-6} \cmidrule(r){7-8}
    Lattice &\multicolumn{2}{c}{Conn.~constant}
    & {Max. degree} 
    & \multicolumn{1}{r}{$\Delta$} & {$\lambda$}
    & {$\Delta$} & {$\lambda$}  \\

    \midrule

    $\mathbb{T}$ & 4.251 419&\cite{alm_upper_2005}
    & 6 & 4.325 & 0.937 &
    5 & 0.762 \cite{Weitz06CountUptoThreshold}\\

    $\mathbb{H}$ & 1.847 760&\cite{duminil-copin_connective_2012}
    & 3 & 1.884 & 4.706 &
    2 & 4.0 \cite{Weitz06CountUptoThreshold}\\

    $\Z^2$ & 2.679 193&\cite{poenitz00:_improv_z}
    & 4 & 2.731 & 2.007 &
    2.502 &
    2.48 \cite{restrepo11:_improv_mixin_condit_grid_count,vera13:_improv_bound_phase_trans_hard}\\

    $\Z^3$ & 4.7387&\cite{poenitz00:_improv_z}
    & 6 & 4.765 & 0.816 &
    5&0.762 \cite{Weitz06CountUptoThreshold}\\

    $\Z^4$  & 6.804 0&\cite{poenitz00:_improv_z}
    & 8 & 6.818 &  0.506 &
    7 & 0.490 \cite{Weitz06CountUptoThreshold}\\

    $\Z^5$ & 8.860 2&\cite{poenitz00:_improv_z}
    & 10 & 8.868 & 0.367 &
    9 & 0.360 \cite{Weitz06CountUptoThreshold}\\

    $\Z^6$ & 10.888 6&\cite{weisstein:_self_avoid_walk_connec_const}
    & 12 & 10.894 & 0.288 &
    11 &   0.285 \cite{Weitz06CountUptoThreshold}\\

    \bottomrule
  \end{tabular}
  \caption{Strong spatial mixing bounds for various lattices.  ($Z^D$
    is the $D$-dimensional Cartesian lattice; $\mathbb{T}$ and $\mathbb{H}$
    denote the triangular and honeycomb lattices respectively.)} 
\ifbool{csub}{%
  \vspace{-2.5\baselineskip}%
}{%
}
  \label{fig:1}
\ifbool{csub}{%
   \end{table*}
}{%
   \end{table}
}

Finally, we apply our techniques to get tighter results for
correlation decay in some other important cases.  The first of these
is strong spatial mixing on \emph{spherically symmetric
trees} (rooted trees in which the degree of each vertex is dependent
only upon its distance from the root).  Such trees have been studied before,
for example by Lyons~\cite{lyons_random_1990}.  For these trees, we
improve upon Theorem~\ref{thm:main-1} to show that strong spatial
mixing holds as long as $\lambda < \lambda_c(\Delta)$,
which is optimal as a function of $\Delta$.  As an application,
we present an FPTAS for the partition function of the hard core model
on bounded degree bipartite graphs that improves upon the parameters
promised by Weitz's algorithm~\cite{Weitz06CountUptoThreshold}; see
\ifbool{csub}{%
  the full version~\cite{sinclair13:_spatial} for details. 
}{%
  Section~\ref{sec:spher-symm-trees} for details.
}%
We also consider the question of uniqueness of Gibbs measure on
general trees.  Uniqueness is a weaker notion than spatial mixing,
requiring correlations to decay to zero with distance but
not necessarily at an exponential rate (see
\ifbool{csub}{%
  the full version for a formal
  definition).  %
}{%
  Section~\ref{sec:branch-numb-uniq} for a formal definition).  %
}  %
We show
that the threshold for uniqueness of the Gibbs measure of the hard
core model on a general tree can be related to the \emph{branching
factor} of the tree, another natural notion of average arity that has
appeared in the study of uniqueness of Gibbs measure for models such
as the ferromagnetic Ising model~\cite{lyons_ising_1989} and is
slightly stronger than the connective constant.
The details of these
results can be found in 
\ifbool{csub}{%
  the full version~\cite{sinclair13:_spatial} of this paper.  %
}{%
  Section~\ref{sec:branch-numb-uniq}.  %
} %

We close this section with a few remarks on our proof techniques,
which begin with the message approach first used by Restrepo~\emph{et
al.}~\cite{restrepo11:_improv_mixin_condit_grid_count}.  The key
idea here is to consider a function, called a
\emph{message}, of the occupation probabilities, and to study how
``errors'' in this message propagate through the standard tree
recurrence.  All previous applications of this
method~\cite{restrepo11:_improv_mixin_condit_grid_count,
sinclair_approximation_2012, li_approximate_2012,
li_correlation_2011} establish correlation decay on the tree by
showing that, for an appropriate range of parameters and for an
appropriately defined message, the error at every vertex in the
tree is at most some
constant fraction of the \emph{maximum} error at its children.
This is an inherently worst case analysis that is oblivious to the
more detailed structure of the tree.  At a very high level, our
main technical innovation gets around this issue by considering
the decay of a smoother statistic (an $\ell_q$ norm for $q < \infty$)
of the errors at the children rather than 
the maximum error.  We present a general, message-independent
framework for achieving this in Section~\ref{sec:messages-tree},
and then in Section~\ref{sec:special-message}
we instantiate the framework for a specific message (previously used
in \cite{li_correlation_2011})
to obtain the proofs of our main results.  We note that our message
framework extends in a straightforward manner to other two-spin
systems such as the anti-ferromagnetic Ising model; this extension is
deferred to a future paper.

\subsection{Related work}
\label{sec:related-work}
Luby and Vigoda~\cite{luby_approximately_1997} were the first to give
a general approximation algorithm (an FPRAS) for the hard core model
partition function valid for $\lambda < \frac{1}{d-3}$ on graphs of
maximum degree $d$.  In a breakthrough result, Weitz~\cite{Weitz06CountUptoThreshold} proved that for
the hard core model, uniqueness of the Gibbs measure on the infinite
$d$-ary tree implies that all graphs of degree at most $d+1$ exhibit
strong spatial mixing, and further that there is an FPTAS for the
partition function of the hard core model on such graphs.  Weitz's
results have since been extended to the anti-ferromagnetic Ising model
with an arbitrary external field~\cite{sinclair_approximation_2012},
and later to general anti-ferromagnetic two-spin
systems~\cite{li_correlation_2011}.  On the other hand, in the special
case of the hard core model on the two dimensional Cartesian lattice
$\Z^2$, Restrepo, Shin, Tetali, Vigoda and Yang~\cite{restrepo11:_improv_mixin_condit_grid_count} improved upon
the threshold for strong spatial mixing obtained using a direct
application of Weitz's result.  Their numerical bound for $\Z^2$ has
recently been improved by Vera, Vigoda and
Yang~\cite{vera13:_improv_bound_phase_trans_hard}, using methods that
are similar in spirit to those in~\cite{restrepo11:_improv_mixin_condit_grid_count}, but which require
even more intensive and sophisticated numerical computations.
As observed earlier, although the
methods of Restrepo \emph{et
  al.}~\cite{restrepo11:_improv_mixin_condit_grid_count} (as well
those of Vera \emph{et al.}~\cite{vera13:_improv_bound_phase_trans_hard})
can in principle be applied to any given lattice, in practice such an
application would require a significant amount of computer-assisted
numeric and symbolic computations, especially for higher dimensional
lattices.  Further, unlike ours, their methods do not seem
to generalize to arbitrary graph families (such as the
random graph family $\G\inp{n,d/n}$).

All of the results above pertain to the bounded degree case.  Li, Lu
and Yin~\cite{li_approximate_2012,li_correlation_2011} extended some
of the above results for two-spin systems to the setting of general graphs without any
bound on the maximum degree, but only in the
regime of parameters for which uniqueness of the Gibbs measure holds on
the $d$-ary tree \emph{for all} $d$.  In the case of the
hard core model, however, the latter condition fails
to hold for any non-trivial vertex activity.  Our results, in
contrast, impose a condition on the activity only in terms of a
natural notion of effective degree.  Further, this condition is
asymptotically optimal as observed earlier.

Mossel and Sly~\cite{mossel_gibbs_2010} consider the question of
sampling from the hard core distribution on unbounded degree graphs
such as ${\cal G}(n, d/n)$ using Glauber dynamics.  Under some
conditions, their results also extend to other families of sparse
graphs with bounded ``tree excess.''  Although the authors do not give
an explicit bound on the vertex activity $\lambda$ under which their
results hold, an examination of the proofs suggests that necessarily
$\lambda < O\inp{1/d^2}$.  Efthymiou~\cite{efthymiou13:_mcmc_g} has recently improved
upon~\cite{mossel_gibbs_2010} by exhibiting rapid mixing for a Gibbs
sampler under the condition $\lambda < 1/(2d)$.  However, both these
bounds still remain far from our bound of $e/d$. Further,
as noted above, the above results do not seem to hold for
all graphs of bounded connective constant.  Hayes and
Vigoda~\cite{hayes_coupling_2006} also considered the question of
sampling from the hard core model on special classes of unbounded
degree graphs.  They showed that for regular graphs on $n$ vertices of
degree $d(n) = \Omega(\log n)$ and of girth greater than $6$, the
Glauber dynamics for the hard core model mixes rapidly for $\lambda <
(1-\epsilon)e/d(n)$ (where $\epsilon$ is an arbitrary positive
constant).  Their results are incomparable to ours; while our
Theorem~\ref{thm:main-1} requires neither the condition that the graph
should be regular nor any lower bounds on its degree or girth, it does
require additional information about the graph in the form of its
connective constant.  However, when the connective constant is
available, then irrespective of the maximum degree of the graph or its
girth, the theorem affords an FPTAS for the partition function.

Much more progress has been made on relating spatial mixing to
notions of average degree in the case of the zero field \emph{ferromagnetic}
Ising model.  Lyons~\cite{lyons_ising_1989} demonstrated that on an
arbitrary tree, the \emph{branching factor} exactly determines the threshold
for uniqueness of the Gibbs measure for this model. For the ferromagnetic Ising model on general graphs, Mossel and
Sly~\cite{mossel_rapid_2009,mossel_exact_2013} proved results
analogous to our Theorem~\ref{thm:main-1}.  However, the arguments of both \cite{lyons_ising_1989}
and \cite{mossel_rapid_2009,mossel_exact_2013} seem to
rely heavily on special properties of the ferromagnetic Ising model,
and do not appear to be easily extensible to the case of ``repulsive''
spin systems such as the hard core model.  Further, an FPRAS for the
partition function of the ferromagnetic Ising model, without any
restrictions on the degree, is already known~\cite{jersin93,goljerpat03}.

In work related to \cite{lyons_ising_1989} above, Pemantle and Steif~\cite{pemantle_robust_1999} define
the notion of a \emph{robust phase transition (RPT)} and relate
the threshold for RPT for various ``symmetric'' models such as the
zero field Potts model and the Heisenberg model on general trees to
the branching factor of the tree.  In the results of both
\cite{lyons_ising_1989} and \cite{pemantle_robust_1999}, an
important ingredient seems to be the existence of a symmetry group on
the set of spins under whose action the underlying measure remains
invariant.  In contrast, in the hard core model, the two possible spin
states of a vertex (``occupied'' and ``unoccupied'') do not admit any
such symmetry.

The first
reference to the connective constant occurs in classical papers by Hammersley
and Morton~\cite{hammersley_poor_1954}, Hammersley and Broadbent
\cite{broadbent_percolation_1957} and
Hammersley~\cite{hammersley_percolation_1957}.  Since then, several
natural combinatorial questions concerning the number and other
properties of self-avoiding walks in various lattices have been
studied in depth; see the monograph of Madras and
Slade~\cite{madras96:_self_avoid_walk} for a survey.  %
Much work has been devoted especially to finding
rigorous upper and lower bounds for the connective constant of
various lattices~\cite{alm_upper_2005, alm_upper_1993,
  jensen_enumeration_2004, kesten_number_1964, poenitz00:_improv_z}.
Heuristic techniques from physics have also been
brought to bear upon this question. For example, Nienhuis~\cite{nienhuis_exact_1982}
conjectured on the basis of heuristic arguments that the
connective constant of the honeycomb lattice $\mathbb{H}$
must be $\sqrt{2+\sqrt{2}}$.  In a celebrated recent breakthrough, Duminil-Copin
and Smirnov~\cite{duminil-copin_connective_2012} rigorously proved
Nienhuis' conjecture.

\section{Preliminaries}
\label{sec:Preliminaries}

\subsection{The hard core model on trees and graphs}
\label{sec:tree-recurrence}
In this section, we introduce some standard notions
associated with the hard core model on trees and
graphs.  Our notation is similar to that used in other recent works on
the subject~\cite{Weitz06CountUptoThreshold,
  restrepo11:_improv_mixin_condit_grid_count,
  sinclair_approximation_2012,li_approximate_2012,li_correlation_2011}.

Given a graph $G = (V,E)$, a \emph{boundary condition} will refer to a
partially specified independent set in $G$.  Formally, a boundary
condition $\sigma = (S, I)$ is a subset $S \subseteq V$ along with an
independent set $I$ on $S$.

\begin{definition}
  \textbf{(Occupation probability and occupation ratio).} Consider the
  hard core model with vertex activity $\lambda > 0$ on a graph $G$,
  and let $v$ be a vertex in $G$.  Given a boundary condition $\sigma = (S,
  I_S)$ on $G$, the \emph{occupation probability} $p_v(\sigma, G)$ at
  the vertex $v$ is the probability that $v$ is included in an independent
  set $I$ sampled according to the hard core distribution conditioned
  on the event that $I$ restricted to $S$ coincides with $I_S$.  The
  \emph{occupation ratio} $R_v(\sigma,G)$ is then defined as
  \begin{displaymath}
    R_v(\sigma,G) = \frac{p_v(\sigma,G)}{1-p_v(\sigma,G)}.
  \end{displaymath}
\end{definition}

In the special case where the graph $G$ is a tree, both the occupation
probability and the occupation ratio admit a simple recurrence in
terms of similar quantities on subtrees.  Formally, let $T$ be an
arbitrary tree rooted at a vertex $\rho$.  Denote the children of
$\rho$ as $\rho_1, \rho_2, \ldots \rho_d$, and let $T_{\rho_i}$ denote
the subtree of $T$ rooted at $\rho_i$.  Let $\sigma$ be an arbitrary
boundary condition on $T$, and let $\sigma_i$ be its restriction to
$T_{\rho_i}$.  We denote by  $R_i$ the occupation ratio $R_{\rho_i}(\sigma_i, T_{\rho_i})$ of the root
$\rho_i$ of the tree $T_{\rho_i}$.  It is well known---and easy to show---that the $R_i$ obey
the following recurrence~(see, for example,
\cite{Weitz06CountUptoThreshold}):
\begin{equation}
  R_{\rho}(\sigma, T) = f_{d,\lambda}(R_1, R_2, \ldots, R_d) \defeq     
  \lambda\prod_{i=1}^d\frac{1}{1+R_i}.\label{eq:2}
\end{equation}
With a slight abuse of notation, we also use the same notation to
denote a symmetric one-argument version of $f_{d,\lambda}$ defined as
$f_{d,\lambda}(x) \defeq f_{d,\lambda}(x, x, \ldots, x).$

In order to analyze the convergence properties of the above
recurrence, it is often convenient to consider the evolution of a
suitable function of the occupation
ratio, called a
\emph{message} (also known as a \emph{statistic} or \emph{potential})~\cite{restrepo11:_improv_mixin_condit_grid_count,
  sinclair_approximation_2012, li_approximate_2012, li_correlation_2011},  as opposed to the occupation ratio itself.
\begin{definition}\label{def:message}\textbf{(Message).}  
  A \emph{message} is a strictly increasing, continuously
  differentiable function $\phi: (0,\infty) \rightarrow \R$,
  such that the derivative of
  $\phi$ on any interval of the form $(0, M]$ is bounded away from
  $0$. 
\end{definition}
\noindent Note that the conditions on $\phi$ imply that it has a
continuously differentiable inverse.

In what follows, given a recurrence $f$ for the quantity $R_v$, we
will denote by $f^\phi$ the recurrence for the quantity
$\phi(R_v)$. Formally, 
\ifbool{csub}{%
$f^\phi(x_1,x_2,\ldots,x_d) \defeq
  \phi\inp{f\inp{\psi(x_1), \psi(x_2),\ldots,\psi(x_d)}},$
}{%
$$f^\phi(x_1,x_2,\ldots,x_d) \defeq
  \phi\inp{f\inp{\psi(x_1), \psi(x_2),\ldots,\psi(x_d)}},$$
}
where $\psi$ is the inverse of $\phi$.  Similarly, for a one-argument
recurrence $f$, $f^\phi (x) \defeq \phi\inp{f\inp{\psi\inp{x}}}.$

\begin{definition}\label{def:strong-spatial-mixing}
  \textbf{(Strong Spatial Mixing~\cite{Weitz06CountUptoThreshold}).}
  The hard core model with a fixed vertex activity $\lambda > 0$ is
  said to exhibit \emph{strong spatial mixing} on a family $\family{F}$ of graphs if
  for any graph $G$ in $\family{F}$, any vertex $v$ in $G$, and any two
  boundary conditions $\sigma$ and $\tau$ on $G$ which differ only at
  a distance of at least $\ell$ from $v$, we have
  \begin{equation*}
    \abs{R_v(\sigma, G) - R_v(\tau, G)} = \exp(-\Omega(\ell)).
  \end{equation*}
  In the definition, the family $\family{F}$ might consist of a single infinite
  graph.
\end{definition}

\subsection{Locally finite trees}
We introduce some notation and terminology for locally finite trees.  Let $T$ be a locally finite, but possibly
infinite, tree rooted at some vertex $\rho$.  For two vertices $u$ and
$v$, the statement %
``$u$ is an ancestor of $v$'' is denoted by $u < v$, while $u \leq v$
will denote the statement that ``either $u = v$ or $u$ is an ancestor
of $v$''.  For example, $\rho \leq v$ for all vertices $v$ in $T$.
For any two vertices $u, v$ in $T$ we denote by $d(u, v)$ the
distance between $u$ and $v$.  Further, $\abs{v} = d(\rho, v)$ denotes
the distance of a vertex $v$ from the root~$\rho$.

We will need the notion of a \emph{cutset} (see, for example,
Lyons~\cite{lyons_ising_1989}) in an infinite tree.
\begin{definition}
  \textbf{(Cutset).} Let $T$ be any locally finite infinite tree
  rooted at a vertex $\rho$.  A \emph{cutset} is a finite set of
  vertices $C$ such that (i) any infinite path starting at $\rho$ must
  intersect $C$; and (ii) no vertex in $C$ is an ancestor of another vertex in $C$.
\end{definition}
\begin{remark}
  Notice that if $C$ is a cutset, then so is the set of children of
  vertices in $C$.
\end{remark}
\noindent For a cutset $C$ we define its distance from the root $\rho$
as the minimum distance between the root $\rho$ and any vertex $v$ in
$C$, and denote this distance by $d(\rho, C)$.  Further, we denote by
$T_{\leq C}$ the restriction of $T$ to vertices which are not
descendants of vertices in $C$, and by $T_{<C}$ the
further restriction of $T_{\leq C}$ to vertices not in $C$.

\subsection{Self-avoiding walks and the connective constant}
\label{sec:connective-constant}
As discussed in the introduction, the connective constant is a natural
notion of ``effective degree'' that has been especially well studied in the
case of lattices.  For a vertex $v$ in a locally finite graph $G$, we
will denote by $N(v,\ell)$ the number of self-avoiding walks of length
$\ell$ starting at $v$.  The connective constant of a graph captures the
rate of growth of $N(v,\ell)$ as a function of $\ell$.
\begin{definition}
  \textbf{(Connective constant: infinite graphs).} Let $G=(V,E)$ be a locally
  finite infinite graph.  The \emph{connective constant}
  $\Delta(G)$ is defined as $\sup_{v\in
    V}\limsup_{\ell\rightarrow\infty}N(v,\ell)^{1/\ell}$.
\end{definition}
\begin{remark}
  For vertex-transitive graphs (such as Cartesian lattices),
  the supremum over $v$ can be removed without changing
  the definition.  Moreover, for such graphs, the $\limsup$ can be
  replaced by a limit.  We use
  $\limsup$ in order to avoid issues about the existence of the
  limit for more general classes of graphs.  
\end{remark}
The definition can be easily extended to finite graphs.  For
algorithmic applications, it is natural to define the connective
constant for a family of graphs parametrized by size.  
\begin{definition}
  \textbf{(Connective constant: finite graphs).} Let $\family{F}$ be a family
  of finite graphs.  The connective constant of $\family{F}$ is at most
  $\Delta$ if $\sum_{i=1}^\ell N(v, i) = O(\Delta^\ell)$.  More
  formally, we say that the connective constant of $\family{F}$ is at most
  $\Delta$ if there exist constants $a$ and $c$ such that for any
  graph $G = (V,E)$ in $\family{F}$ and any vertex $v$ in $G$, we have
  $\sum_{i=1}^\ell N(v, i) \leq c\Delta^\ell$ for all $\ell \geq a\log
  |V|$.
\end{definition}
Note that the connective constant of a graph of maximum degree $d+1$
is at most $d$.  However, the connective constant can be much smaller
than the maximum degree.  For example, the maximum degree of a graph
drawn from $\G\inp{n, d/n}$ is $\Theta(\log n/\log\log n)$ with high probability;  however,
it is easy to show (as we do in the proof of
Theorem~\ref{thm:gndn-ssm-alg} below) that for any fixed $\epsilon >
0$, the connective constant is at most $d(1+\epsilon)$ with high
probability.  Similarly, for the 2-dimensional Cartesian lattice
(which has maximum degree $4$) the connective constant is less than
2.68~\cite{alm_upper_2005,poenitz00:_improv_z}.

The set of self-avoiding walks starting at a vertex $v$ in a graph $G$ can be
naturally represented as a tree rooted at $v$, where each vertex in
the tree at distance $\ell$ from $v$ corresponds to a distinct self-avoiding
walk of length $\ell$. Any vertex $u$ in this tree can also be naturally
identified (many-to-one) with the vertex in $G$ at which the corresponding self-avoiding walk
ends.

Weitz~\cite{Weitz06CountUptoThreshold} showed that by fixing certain
vertices of this tree to be occupied or unoccupied, one obtains a tree
$\Ts{v, G}$ (which we will often refer to as the ``Weitz SAW tree'')
such that, for any boundary condition $\sigma$ in $G$, one has
\begin{displaymath}
  R_v(\sigma, G) = R_v(\sigma, \Ts{v, G}),
\end{displaymath}
where on the right hand side, by a slight abuse of notation, we denote
again by $\sigma$ the natural translation of the boundary condition
$\sigma$ on $G$ to $\Ts{v,G}$.  This result forms the cornerstone of
all recent results that use correlation decay for two-spin system on
trees to derive results for other
graphs~\cite{Weitz06CountUptoThreshold,
  restrepo11:_improv_mixin_condit_grid_count,
  sinclair_approximation_2012, li_approximate_2012,
  li_correlation_2011}.  The definition of $\Ts{v,G}$ implies that its
connective constant is always upper-bounded by the connective constant
of $G$. However, it can also be lower than that of $G$ because of the
additional boundary conditions introduced in Weitz's construction.

\ifbool{csub}{
  \paragraph{Notation}
}{%
  \paragraph{Notation.}
}%
For any bivariate function $g(x,y)$, we will
denote the partial derivative $\frac{\partial^{i+j}g}{\partial
  x^i\partial y^j}$ evaluated at $x = a, y = b$ as $g^{(i,j)}(a,b)$. 

\section{Messages on a tree}
\label{sec:messages-tree}

In this section, we study the behavior of a general recurrence $f$ on
a tree.  As before, $T$ is a tree rooted at a vertex $\rho$, and
the occupation ratios $R_\rho$ and $R_i$ are defined in Section~\ref{sec:tree-recurrence}.  We
begin with a version of the
mean value theorem adapted to our setting.  Given a message $\phi$ (as
in Definition~\ref{def:message}), let $\Phi = \phi'$ denote
the derivative of $\phi$; notice that $\Phi(x) > 0$ for all
non-negative $x$ since $\phi$ is strictly increasing.
\begin{lemma}\textbf{(Mean value theorem).}
  \label{lem:mean-value}
  Consider two vectors $\vec{x}$ and $\vec{y}$ in $\phi([0,
  \infty))^d$.  Then there exists a vector $\vec{z} \in [0,
  \infty)^d$ such that
  \begin{displaymath}
    \abs{f_{d,\lambda}^\phi(\vec{x}) - f_{d,\lambda}^\phi(\vec{y})} \leq
    \Phi\inp{f_{d,\lambda}(\vec{z})}\sum_{i=1}^d\frac{\abs{y_i-x_i}}{\Phi(z_i)}\abs{\pdiff{f_{d,\lambda}}{z_i}},
  \end{displaymath}
  where by a slight abuse of notation we denote by
  $\pdiff{f_{d,\lambda}}{z_i}$ the partial derivative of
  $f_{d,\lambda}(R_1, R_2,\ldots,R_d)$ with respect to $R_i$ evaluated
  at $\vec{R} = \vec{z}$.
\end{lemma}
\noindent We defer the proof of this lemma to
\ifbool{csub}{%
  the full version~\cite{sinclair13:_spatial}. %
}{%
  Appendix~\ref{sec:proof-lemma-refl}. %
} %
For the special case of the tree recurrence of the hard core
model~(eq. \ref{eq:2}), Lemma~\ref{lem:mean-value} implies that 
\begin{equation}
  \abs{f_{d,\lambda}^\phi(\vec{x}) - f_{d,\lambda}^\phi(\vec{y})} \leq
  f_{d,\lambda}(\vec{z})\Phi\inp{f_{d,\lambda}(\vec{z})} \sum_{i=1}^d 
  \frac{\abs{y_i-x_i}}{(1+z_i)\Phi(z_i)}.\label{eq:10}
\end{equation}
The first step of our approach is similar to that taken in the
papers~\cite{restrepo11:_improv_mixin_condit_grid_count,
  sinclair_approximation_2012,
  li_approximate_2012,li_correlation_2011} in that we will use an
appropriate message---along with the estimate in
Lemma~\ref{lem:mean-value}---to argue that the ``distance'' between two
input message vectors $\vec{x}$ and $\vec{y}$ at the children of a
vertex shrinks by a constant factor at each step of the recurrence.
Previous works on the
subject~\cite{restrepo11:_improv_mixin_condit_grid_count,
  sinclair_approximation_2012, li_approximate_2012,
  li_correlation_2011} show such a decay on some version of the
$\ell_\infty$ norm of the ``error'' vector $\vec{x} - \vec{y}$: this
is achieved by bounding the appropriate dual $\ell_1$ norm of the
gradient of the recurrence.  Our intuition is that in order to achieve
a bound in terms of a global quantity such as the connective constant, it
should be advantageous to use a more global measure of the error such
as an $\ell_q$ norm
\ifbool{csub}{%
  for $q < \infty$. %
}{%
, for some $q < \infty$. %
}%

In line with the above plan, we first prove the following lemma,
specialized here to the case of the hard core model.  For ease of
notation, we assume throughout that the vertex activity $\lambda > 0$ is fixed
and suppress dependence on $\lambda$.
\begin{lemma}
  \label{lem:tech}
  Let $\phi$ be a message and let $\Phi = \phi'$ be its derivative.
  Let $p$ and $q$ be positive reals such that $\frac{1}{p} +
  \frac{1}{q} = 1$.  Define the functions $S_{\phi, p}$ and
  $\Xi_{\phi,q}(d, x)$ as follows:
  \ifbool{csub}{%
    \begin{displaymath}
      S_{\phi, p}(x) \!\defeq\! \frac{e^{-px}}{\Phi\inp{e^{x}\!-\!1}^p};
      \;  \Xi_{\phi,q}(d, x) \!\defeq\!
      \frac{1}{d}\inp{\frac{d\Phi(f_d(x))f_d(x)}{(1\!+\!x)\Phi(x)}}^q.
    \end{displaymath}%
  }{%
    \begin{displaymath}
      S_{\phi, p}(x) \defeq \inp{\frac{e^{-x}}{\Phi\inp{e^{x} - 1}}}^p;
      \qquad  \Xi_{\phi,q}(d, x) \defeq
      d^{q-1}\inp{\frac{\Phi(f_d(x))f_d(x)}{(1+x)\Phi(x)}}^q.
    \end{displaymath}%
  }%
  We further define $\xi_{\phi,q}(d) \defeq \sup_{x \geq 0}\Xi_{\phi,q}(d,
  x)$.  If $S_{\phi,p}$ is a \emph{concave} function on the non-negative
  reals, then for any two vectors $\vec{x}, \vec{y}$ in
  $\phi(\R^+)^d$, we have
  \begin{displaymath}
    \abs{f_d^\phi(\vec{x}) - f_d^\phi(\vec{y})}^q \leq
    {\xi_{\phi,q}(d)}\norm[q]{\vec{x} - \vec{y}}^q.
  \end{displaymath}
\end{lemma}
\begin{proof}
  The concavity of $S_{\phi,p}(x)$ for
  non-negative $x$, combined with Jensen's inequality, implies that for
  any vector $\vec{z} \in [0,\infty)^d$, and $Z=\prod_{i=1}^d(1+z_i)^{1/d}-1$,
\ifbool{csub}{%
  \begin{align}
    \frac{1}{d}\sum_{i=1}^d\inp{\frac{1}{(1+z_i)\Phi(z_i)}}^p 
    &=
    \frac{1}{d}\sum_{i=1}^d S_{\phi,p}\inp{\ln(1+z_i)}\nonumber\\
    &\leq
    S_{\phi,p}\inp{\frac{1}{d}\sum_{i=1}^d\ln(1+z_i) }\nonumber\\
    &=
    \inp{\frac{1}{(1+Z)\Phi(Z)}}^p.\label{eq:Jensen}
  \end{align}
}{%
  \begin{equation}
    \frac{1}{d}\sum_{i=1}^d\inp{\frac{1}{(1+z_i)\Phi(z_i)}}^p 
    =
    \frac{1}{d}\sum_{i=1}^d S_{\phi,p}\inp{\ln(1+z_i)}
    \leq
    S_{\phi,p}\inp{\frac{1}{d}\sum_{i=1}^d\ln(1+z_i) }
    =
    \inp{\frac{1}{(1+Z)\Phi(Z)}}^p.\label{eq:Jensen}
  \end{equation}%
}%
  Besides, it is easy to verify that $f_{d}(Z) = f_{d}(\vec{z})$.
  Now, we apply Lemma\nobreakspace \ref {lem:mean-value}
  (specifically, its consequence in eq.~(\ref{eq:10})).  Assume that $\vec{z}$ is as
  defined in that lemma, and again let $Z$ denote $\prod_{i=1}^d(1+z_i)^{1/d}-1$.  We
  then have 
\ifbool{csub}{%
  \begin{align}
    \abs{f_d^\phi(\vec{x}) \!-\! f_d^\phi(\vec{y})} &=
    \Phi(f_d(\vec{z}))f_d(\vec{z})
    \sum_{i=1}^d\frac{\abs{y_i -
        x_i}}{(1+z_i)\Phi(z_i)}\nonumber\\
    &\leq d^{1/p}\Phi(f_d(Z))f_d(Z) \nonumber\\
    &\quad\cdot\inp{\frac{1}{d}\sum_{i=1}^d
      \inp{\frac{1}{(1\!+\!z_i)\Phi(z_i)}}^p}^{1/p}\norm[q]{\vec{x}\!-\!\vec{y}},\nonumber\\
    &\leq
    \frac{d^{1/p}\Phi\inp{f_d(Z)}f_d(Z)}{(1\!+\!Z)\Phi\inp{Z}}\norm[q]{\vec{x}\!-\!\vec{y}},\nonumber
  \end{align}%
  where we use eq.~(\ref{eq:10}) in the first relation, $f_{d}(Z) =
  f_{d}(\vec{z})$ and Hölder's inequality in the second relation, and
  eq.\nobreakspace \textup {(\ref {eq:Jensen})} in the third
  relation. %
}{%
  \begin{align}
    \abs{f_d^\phi(\vec{x}) - f_d^\phi(\vec{y})} &=
    \Phi(f_d(\vec{z}))f_d(\vec{z})
    \sum_{i=1}^d\frac{\abs{y_i -
        x_i}}{(1+z_i)\Phi(z_i)}\qquad\text{using eq.~(\ref{eq:10})}\nonumber\\
    &\leq d^{1/p}\Phi(f_d(Z))f_d(Z) \inp{\frac{1}{d}\sum_{i=1}^d
      \inp{\frac{1}{(1+z_i)\Phi(z_i)}}^p}^{1/p}\norm[q]{\vec{x} - \vec{y}},\nonumber\\
    &\qquad\text{using $f_{d}(Z) = f_{d}(\vec{z})$ and
      Hölder's inequality}
    \nonumber\\
    &\leq
    \frac{d^{1/p}\Phi\inp{f_d(Z)}f_d(Z)}{(1+Z)\Phi\inp{Z}}\norm[q]{\vec{x}-\vec{y}},
    \text{ using eq.\nobreakspace \textup {(\ref {eq:Jensen})}}.\nonumber
  \end{align}%
}%
Raising both sides to the $q$th power, using $\frac{1}{p} +
\frac{1}{q} = 1$, and the definitions of the function $\Xi$ and $\xi$,
we get the claimed inequality.
\end{proof}
Let $\phi$ be a message satisfying the conditions of Lemma\nobreakspace \ref {lem:tech}.
Lemma\nobreakspace \ref {lem:tech} then implies the following general lemma on
propagation of ``errors'' in locally finite infinite trees.  As
before, we denote by $R_\rho(\sigma)$ the occupation probability of
the root $\rho$ in the hard core model with boundary condition
$\sigma$.  We consider the dependence of $R_\rho(\sigma)$ on the
boundary conditions $\sigma$ which are fixed everywhere except at some
cutset $C$.  For technical reasons, we will assume that the conditions
are actually allowed to differ not on the cutset $C$ itself, but on
the cutset $C'$ which is the set of children of the vertices in
$C$.
\begin{lemma}
  Let $T$ be a locally finite tree rooted at $\rho$. Let $C$ be a
  cutset in $T$ at distance at least $1$ from the root, and $C'$ the
  cutset comprising of all children of vertices in $C$.  Consider two
  arbitrary boundary conditions $\sigma$ and $\tau$ on $T_{\leq C'}$
  which differ only on $C'$.  If $\phi$ and $q$ satisfy the conditions
  of Lemma\nobreakspace \ref {lem:tech}, we have
  \begin{equation*}
    |R_\rho(\sigma) - R_\rho(\tau)|^q \leq \frac{c_0\xi_{\phi,q}(d_\rho)}{\alpha'}
    \sum_{v \in C} \alpha^{|v|},
  \end{equation*}
  where $d_\rho$ is the degree of the root $\rho$, $c_0$ is a constant
  depending only upon $q$, $\lambda$ and the message $\phi$, while
  $\alpha = \sup \xi_{\phi,q}(d)$, where the supremum is taken over
  the arities $d$ of all vertices in $T$ except the root $\rho$, and
  $\alpha' = \inf_{d \geq 1} \xi_{\phi,q}(d)$. 
  \label{lem:general-tree}
\end{lemma}
For a proof of this lemma, see 
\ifbool{csub}{%
  the full version~\cite{sinclair13:_spatial}. %
}{%
  Appendix~\ref{sec:proof-lemma-refl}. %
}

\section{A special message}
\label{sec:special-message}

We now instantiate the approach outlined in
Section~\ref{sec:messages-tree} to prove Theorems\nobreakspace \ref {thm:main-1} and\nobreakspace  \ref {thm:main-2}.
Our message is the same as that
used in~\cite{li_correlation_2011}; %
\ifbool{csub}{%
  we define
  \begin{equation}
    \phi(x) \defeq \sinh^{-1}\sqrt{x}\text{; }\Phi(x) \defeq
    \phi'(x) = \frac{1}{2\sqrt{x(1+x)}}.\label{eq:4}
  \end{equation}%
}{%
  we choose
  \begin{equation}
    \phi(x) \defeq \sinh^{-1}\inp{\sqrt{x}}\text{, so that }\Phi(x) \defeq
    \phi'(x) = \frac{1}{2\sqrt{x(1+x)}}.\label{eq:4}
  \end{equation}%
}%
Notice that $\phi$ is a strictly increasing, continuously
differentiable function on $(0, \infty)$, and also satisfies the
technical condition on the derivative $\Phi$ as required in the
definition of a message. 
Moreover, we choose $p=q=2$.
Our choices are
designed to satisfy the conditions of Lemma\nobreakspace \ref {lem:tech}.  We now proceed
to show that this is indeed the case.
\begin{observation}
  The function $S_{\phi,2}(x)$, when $\phi$ is as defined in eq.\nobreakspace \textup {(\ref {eq:4})},
  is concave for $x \geq 0$.
\end{observation}
\begin{proof}
  For $x \geq 0$, and $\Phi$ defined in eq.\nobreakspace \textup {(\ref {eq:4})} we have
  $S_{\phi, 2}(x) = 4(1 - e^{-x})$, which is concave.
\end{proof}
In order to see what Lemma\nobreakspace \ref {lem:tech} implies with
our message, we first analyze the function $\xi(d)\defeq\xi_{\phi,2}(d)$ in
some detail.  The proof of the following simple lemma follows from 
arguments in \cite{li_correlation_2011};  however, we include a
proof here for completeness.  In what follows, we drop the subscripts
$\phi$ and $q$ since these
will always be clear from the context.
\begin{lemma}
  \label{lem:fixed-point}
  Consider the hard core model with any fixed vertex activity $\lambda
  > 0$. With $\phi$ as defined in eq.\nobreakspace \textup {(\ref {eq:4})}, 
  we have $\xi(d) =
  \Xi(d, \tilde{x}_\lambda(d))$, where $\tilde{x}_\lambda(d)$ is the unique
  solution to
  \begin{equation}
    d\tilde{x}_\lambda(d) = 1 + f_{d,\lambda}(\tilde{x}_\lambda(d)). \label{eq:6}
  \end{equation}
\end{lemma}
\begin{proof}
Plugging in $\Phi$ from eq.\nobreakspace \textup {(\ref {eq:4})} in the definition of $\Xi$ and
using $q = 2$, we get 
\begin{displaymath}
  \Xi(d, x) = \frac{dx}{1+x}\frac{f_{d,\lambda}(x)}{1+f_{d,\lambda}(x)}.
\end{displaymath}
Taking the partial derivative with respect to the second argument, we
get 
\begin{displaymath}
  \Xi^{(0,1)}(d,x) = \frac{\Xi(d,x)}{x(1+x)\inp{1+f_{d,\lambda}(x)}}
  \insq{1 + f_{d,\lambda}(x)  - dx}. 
\end{displaymath}
For fixed $d$, the quantity outside the square brackets is always
positive, while the expression inside the square brackets is strictly
decreasing in $x$.  Thus, any zero of the expression in the brackets
will be a unique maximum of $\Xi$.  The fact that such a zero exists
follows by noting that the partial derivative is positive at $x = 0$
and negative as $x \rightarrow \infty$.  Thus, $\Xi(d, x)$ is
maximized at $\tilde{x}_{\lambda}(d)$ as defined above, and hence
$\xi(d) = \Xi(d, \tilde{x}_\lambda(d))$, as claimed.
\end{proof}
We now define $\nu_{\lambda}(d) \defeq \xi(d)$.  We first
show that this definition agrees with the one used in the statement of
Theorem\nobreakspace \ref {thm:main-2}, and then derive some monotonicity properties of the
function $\nu$.
\begin{lemma}\label{lem:tau-props}
  For a given $\lambda > 0$ and a positive integer $d$ let $\tilde{x}_\lambda(d)$
  be the unique solution to eq.\nobreakspace \textup {(\ref {eq:6})}.  We then have
  \begin{enumerate}
  \item $\nu_\lambda(d) = \frac{d\tilde{x}_{\lambda}(d) - 1}{1 +
      \tilde{x}_{\lambda}(d)}$ and $\nu_\lambda(d) = \frac{1}{d}$
    when $\lambda = \lambda_c(d)$.\label{item:1}
  \item $\nu_\lambda(d)$ is increasing in $d$ for fixed $\lambda > 0$.\label{item:2}
  \item $\nu_\lambda(d)$ is increasing in $\lambda$ for fixed $d >
    0$.\label{item:3}
  \end{enumerate}
\end{lemma}
The proof of the above lemma %
\ifbool{csub}{%
is deferred to the full version~\cite{sinclair13:_spatial}. %
}{%
is somewhat technical, and is deferred to Appendix~\ref{sec:monot-prop-nu_l}. %
} %
We now proceed with the proof of Theorem\nobreakspace \ref {thm:main-2}.  
\begin{proof}[Proof of Theorem\nobreakspace \ref {thm:main-2}] 
  Let $\family{F}$ be any family of finite or infinite graphs with connective constant $\Delta$
  and maximum degree $d+1$.  We prove the result for any fixed
  $\lambda$ such that $\nu_\lambda(d)\Delta = 1 - \epsilon$ for some
  fixed $\epsilon > 0$.  The result will then follow for all $\lambda'
  \leq \lambda$, since by item\nobreakspace \ref {item:3} in Lemma\nobreakspace \ref {lem:tau-props} we
  then have $\nu_{\lambda'}(d)\Delta \leq 1 - \epsilon$ for every
  such $\lambda'$.

  We first prove that the hard core model with these parameters
  exhibits strong spatial mixing on this family of graphs. Let $G$ be
  any graph from $\family{F}$, $v$ any vertex in $G$, and consider
  any boundary conditions $\sigma$ and $\tau$ on $G$ which differ only
  at a distance of at least $\ell$ from $v$.  We consider the Weitz
  self-avoiding walk tree $\Ts{v, G}$
  rooted at $v$ (as defined in Section~\ref{sec:connective-constant}).
   As before, we denote again by $\sigma$ (respectively, $\tau$)
   the translation of the boundary condition $\sigma$ (respectively,
   $\tau)$ on $G$ to $\Ts{v, G}$.  From Weitz's theorem, we then have that
  $R_v(\sigma, G) = R_v(\sigma, \Ts{v, G})$ (respectively,
  $R_v(\tau, G) = R_v(\tau, \Ts{v, G})$).  

  Consider first the case where
  $G$ is infinite.  Let $C_\ell$ denote the cutset in $\Ts{v, G}$ consisting of all
  vertices at distance $\ell$ from $v$. Since $G$ has connective constant at most $\Delta$,
  it follows that for $\ell$ large enough, we have $|C_\ell| \leq
  \Delta^\ell(1-\epsilon/2)^{-\ell}$.  Notice that the maximum degree of
  $\Ts{v, G}$ is $d+1$, and hence every vertex except for the root has
  arity at most $d$.  In the notation of Lemma\nobreakspace \ref {lem:general-tree},
  $\nu_\lambda(d)\Delta = 1-\epsilon$ then implies that $\alpha \leq
  (1-\epsilon)/\Delta$.  Further, since $\nu$ is increasing in $d$, we
  also have $\alpha' \defeq \inf_{d \geq 1} \nu(d) = \nu(1) > 0$.  Applying Lemma\nobreakspace \ref {lem:general-tree} we then get
\ifbool{csub}{%
  \begin{align*}
    \abs{R_v(\sigma, G) - R_v(\tau, G)}^2 &= \left|R_v(\sigma, \Ts{v,G})\right.\\
      &\qquad\qquad\left. - R_v(\tau, \Ts{v,G})\right|^2\\
    &\leq \frac{c_0}{\nu_\lambda(1)}\nu_\lambda(d_v)\sum_{u \in
      C_\ell}\inp{\frac{1-\epsilon}{\Delta}}^\ell\\
    &\leq \frac{c_0}{\nu_\lambda(1)}\nu_\lambda(d+1)
    \inp{\frac{1-\epsilon}{1-\epsilon/2}}^\ell,
  \end{align*}%
  where, in the second line, $d_v$ denotes the degree of the vertex
  $v$ in $G$, and in the third line we use $|C_\ell| \leq
  \Delta^\ell(1-\epsilon/2)^{-\ell}$ and $\nu_\lambda(d_v) \leq
  \nu_\lambda(d+1)$. This establishes strong spatial mixing in $G$, since
  $1-\epsilon < 1-\epsilon/2$. 

}{%
  \begin{align*}
    \abs{R_v(\sigma, G) - R_v(\tau, G)}^2 &= \abs{R_v(\sigma, \Ts{v,G})
      - R_v(\tau, \Ts{v,G})}^2\\
    &\leq \frac{c_0}{\nu_\lambda(1)}\nu_\lambda(d_v)\sum_{u \in
      C_\ell}\inp{\frac{1-\epsilon}{\Delta}}^\ell\text{, where $d_v$
      is
      the degree of $v$ in $G$}\\
    &\leq \frac{c_0}{\nu_\lambda(1)}\nu_\lambda(d+1)
    \inp{\frac{1-\epsilon}{1-\epsilon/2}}^\ell,\\
    &\qquad\text{ using $|C_\ell| \leq \Delta^\ell(1-\epsilon/2)^{-\ell}$ and
      $\nu_\lambda(d_v) \leq \nu_\lambda(d+1)$,}
  \end{align*}%
  which establishes strong spatial mixing in $G$, since
  $1-\epsilon < 1-\epsilon/2$. 
}

  We now consider the case when $\family{F}$ is a family of finite graphs, and
  $G$ is a graph from $\family{F}$ of $n$ vertices.  Since the connective constant
  of the family is $\Delta$, there exist constants $a$ and $c$ such that for
  $\ell \geq a \log n$, $\sum_{i=1}^\ell N(v, \ell) \leq c\Delta^\ell$.  We now proceed
  with the same argument as in the infinite case, but choosing $\ell
  \geq a \log n$.  The cutset $C_\ell$ is again chosen to be the set
  of all vertices at distance $\ell$ from $v$ in $\Ts{v,G}$, so that
  $|C_\ell| \leq c\Delta^{\ell}$.  As before, we then have for $\ell >
  a \log n$,
\ifbool{csub}{%
  \begin{align}
    \abs{R_v(\sigma, G) - R_v(\tau, G)}^2 &= \left|R_v(\sigma,
      \Ts{v,G})\right.\nonumber\\
     &\qquad\qquad \left.- R_v(\tau, \Ts{v,G})\right|^2\nonumber\\
    &\leq \frac{c_0}{\nu_\lambda(1)}\nu_\lambda(d_v)\sum_{u \in
      C_\ell}\inp{\frac{1-\epsilon}{\Delta}}^\ell\nonumber\\
    &\leq \frac{cc_0}{\nu_\lambda(1)}\nu_\lambda(d\!+\!1)\inp{1\!-\!\epsilon}^{\ell},\label{eq:12}
  \end{align}%
  where, in the second relation, $d_v$ is the degree of the vertex $v$
  in $G$, and in the third relation we use $|C_\ell| \leq c\Delta^\ell$
  and $\nu_\lambda(d_v) \leq \nu_\lambda(d+1)$.  This establishes the
  requisite strong spatial mixing bound.
}{%
  \begin{align}
    \abs{R_v(\sigma, G) - R_v(\tau, G)}^2 &= \abs{R_v(\sigma, \Ts{v,G})
      - R_v(\tau, \Ts{v,G})}^2\nonumber\\
    &\leq \frac{c_0}{\nu_\lambda(1)}\nu_\lambda(d_v)\sum_{u \in
      C_\ell}\inp{\frac{1-\epsilon}{\Delta}}^\ell\text{, where $d_v$
      is
      the degree of $v$ in $G$}\nonumber\\
    &\leq c\frac{c_0}{\nu_\lambda(1)}\nu_\lambda(d+1)\inp{1-\epsilon}^{\ell}
    \text{, using $|C_\ell| \leq c\Delta^\ell$ and $\nu_\lambda(d_v) \leq
      \nu_\lambda(d+1)$,}\label{eq:12}
  \end{align}%
  which establishes the requisite strong spatial mixing bound. 
}

  In order to prove the algorithmic part, we first recall an observation of
  Weitz~\cite{Weitz06CountUptoThreshold} that an FPTAS for the
  ``non-occupation'' probabilities $1 - p_v$ under arbitrary boundary
  conditions is sufficient to derive an FPTAS for the partition
  function.  We further note that if the vertex $v$ is not already
  fixed by a boundary condition, then  $1-p_v = \frac{1}{1 + R_v} \geq
  \frac{1}{1+\lambda}$, since $R_v$ lies in the interval $[0,\lambda]$
  for any such vertex. Hence, an additive approximation to $R_v$ with
  error $\mu$ implies a multiplicative approximation to $1-p_v$ within
  a factor of $1\pm \mu(1+\lambda)$.  Thus, an algorithm that
  produces in time polynomial in $n$ and $1/\mu$ an \emph{additive}
  approximation to $R_v$ with error at most $\mu$ immediately gives an
  FPTAS for $1-p_v$, and hence, by Weitz's observation, also for the
  partition function. %
  To derive such an algorithm, we again use the tree $\Ts{v, G}$ considered
  above.  Suppose we require an additive approximation with error at
  most $\mu$ to $R_v(\sigma, G) = R_v(\sigma, \Ts{v, G})$.  We notice
  first that $R_v = 0$ if and only if there is a neighbor of $v$ that
  is fixed to be occupied in the boundary condition $\sigma$.  In this
  case, we simply return $0$.  Otherwise, we expand $\Ts{v, G}$ up to
  depth $\ell$ for some $\ell \geq a\log n$ to be specified later.
  Notice that this subtree can be explored in time
  $O\inp{\sum_{i=1}^\ell N(v,i)}$ which is $O(\Delta^\ell)$ since the
  connective constant is at most $\Delta$.

  We
  now consider two extreme boundary conditions $\sigma_+$ and
  $\sigma_-$ on $C_\ell$: in $\sigma_+$ (respectively, $\sigma_-$) all
  vertices in $C_\ell$ that are not already fixed by $\sigma$ are
  fixed to ``occupied'' (respectively, unoccupied).  The form of the
  recurrence ensures that the true value $R_v(\sigma)$ lies between
  the values $R_v(\sigma_+)$ and $R_v(\sigma_-)$.  We compute the
  recurrence for both these boundary conditions on the tree.  The
  analysis leading to eq.\nobreakspace \textup {(\ref {eq:12})}
  ensures that, since $\ell \geq a \log n$, we have
  \begin{displaymath}
    \abs{R_v(\sigma_+, G) - R_v(\sigma_-, G)} \leq M_1\exp(-M_2\ell)
  \end{displaymath}
  for some fixed positive constants $M_1$ and $M_2$. Now, assume
  without loss of generality that $R_v(\sigma_+) \geq R_v(\sigma_-)$.
  By the preceding observations, we then have $$R_v(\sigma) \leq
  R_v(\sigma_+) \leq R_v(\sigma) + M_1\exp(-M_2\ell).$$ By choosing
  $\ell = a\log n + O(1) +O(\log (1/\mu))$, we get the required
  $\pm\mu$ approximation.  Further, by the observation above, the
  algorithm runs in time $O\inp{\Delta^{\ell}}$, which is polynomial in
  $n$ and $1/\mu$ as required.
\end{proof}

We now prove Theorem\nobreakspace \ref {thm:main-1}.  
\begin{proof}[Proof of Theorem\nobreakspace \ref {thm:main-1}]
  The spatial mixing part of the Theorem follows directly from
  Theorem\nobreakspace \ref {thm:main-2} once we prove that 
  $\lambda < \lambda_c(\Delta+1)$ implies that $\sup_{d\geq 1}\nu_\lambda(d)\Delta < 1$.  We now proceed with this verification.  For
  ease of notation, we will denote $\tilde{x}_\lambda(d)$ (as
  defined by eq.~\eqref{eq:6}) by $\tilde{x}$ in the proof.

  Let $D$ be such that $\lambda = \lambda_c(D)$.  Since
  $\lambda < \lambda_c(\Delta + 1)$, we have $D > \Delta + 1 > 1$, and
  hence there exists an $\epsilon > 0$ such that  $\nu_\lambda(d)\Delta \leq \frac{1}{1+\epsilon}\nu_\lambda(d)(D-1)$ for all $d
  \geq 1$.  Thus, in order to establish that $\sup_{d\geq
    1}\nu_\lambda(d)\Delta < 1$, we only need to show that $\nu_\lambda(d)(D-1) < 1$
  for all $d \geq 1$. 
  Using the definition $\nu_\lambda(d) =
  \frac{d\tilde{x} -1}{1+\tilde{x}}$, this translates into the
\ifbool{csub}{%
  requirement  $(d(D-1) - 1)\tilde{x} < D$ for $\tilde{x}$.
}{%
  following requirement for $\tilde{x}$:
    \begin{displaymath}
      (d(D-1) - 1)\tilde{x} < D.
    \end{displaymath}
} %
    Now, the definition of $\tilde{x}$ implies that $\tilde{x} \geq
    1/d > 0$, so this condition is trivially satisfied if
    $d(D-1) \leq 1$.  Hence, we assume from now on that $d(D-1) - 1 >
    0$.  In this case, the above condition translates to 
    \begin{displaymath}
      \tilde{x} < x^\star \defeq \frac{D}{d(D-1) - 1}.
    \end{displaymath}
    We now notice that eq.~(\ref{eq:6}) for $\tilde{x}$ can be written as
    $P_{d,D}(\tilde{x}) = 0$, where $P_{d,D}(x) \defeq dx - 1 -
    f_{d,\lambda_c(D)}(x)$ is a strictly increasing function of $x$.
    Thus the above requirement translates to $P_{d,D}(x^\star) > 0$,
    which simplifies to the requirement
\ifbool{csub}{%
    \begin{equation}
      \label{eq:14}
      \insq{1 + \frac{D}{d(D\!-\!1) - 1}}^{d\!+\!1} >
      \inp{\frac{D}{D\!-\!1}}^D\text{, for all $d \geq 1$}.
    \end{equation} %
}{%
    \begin{equation}
      \label{eq:14}
      \insq{1 + \frac{D}{d(D-1) - 1}}^{d+1} >
      \inp{\frac{D}{D-1}}^D\quad\text{ for all $d \geq 1$}.
    \end{equation} %
} %
    Since the left hand side of eq.\nobreakspace \textup {(\ref {eq:14})} is a
    decreasing function of $d$, we only need to verify the
    condition in the limit $d \rightarrow \infty$.  This is equivalent
    to the condition
\ifbool{csub}{%
    $e^{{D}/\inp{D-1}} > {D^D}/{(D-1)^D}$, %
}{%
    \begin{equation*}
      e^{\frac{D}{D-1}} > \inp{\frac{D}{D-1}}^D,
    \end{equation*} %
}%
    which in turn is equivalent to the inequality $e^{1/(D-1)} >
    1 + \nfrac{1}{(D-1)}$, and the latter is true for all $D > 1$.  Thus we
    see that whenever $\Delta + 1 < D$, or equivalently, when
    $\lambda_c(D) = \lambda < \lambda_c(\Delta + 1)$, we have
    $\nu_\lambda(d)\Delta < 1$.

    In order to prove the claim that $\lambda_c(\Delta+1) \geq
    \frac{e}{\Delta}$ in the remark following
    Theorem~\ref{thm:main-1}, we notice that
    \begin{displaymath}
      \lambda_c(\Delta+1) =
      \frac{1}{\Delta}\inp{1+\frac{1}{\Delta}}^{\Delta+1} \geq \frac{e}{\Delta},
    \end{displaymath}
    since $(1+1/x)^{x+1} > e$ for $x \geq 0$. 
    
    The proof of the algorithmic part of the theorem is identical to
    the proof of the algorithmic part of Theorem\nobreakspace \ref
    {thm:main-2}. As before, the depth $\ell$ up to which the
    self-avoiding walk tree needs to be explored in order to 
    additively approximate $R_v$ up to an error of at most $\mu$ is
    $a\log n + O(\log (1/\mu)) = O(\log n + \log (1/\mu))$, and
    hence the running time $O\inp{\Delta^{\ell}}$ is still polynomial in
    $n$ and $1/\mu$.
\end{proof}

Finally, we apply Theorem~\ref{thm:main-1} to the random graph model ${\cal G}(n, d/n)$, and prove
Theorem~\ref{thm:gndn-ssm-alg}.  We first restate
Theorem~\ref{thm:gndn-ssm-alg} to make the statements about the
various high probability events more precise. 
{
\begingroup
\def\thetheorem{\ref*{thm:gndn-ssm-alg}}
\begin{theorem}\textbf{\emph{(Restated).}}
  Let $\epsilon$, $\beta > 0$, $d > 1$ and $\lambda <
  \frac{e}{d(1+\epsilon)}$ be fixed. Then:
  \begin{itemize}
  \item The hard core model with activity $\lambda$ on a graph $G$
    drawn from $\G(n,d/n)$ exhibits strong spatial mixing with
    probability at least $1 - n^{-\beta}$. 
  \item There exists a deterministic algorithm which, on input $\mu >
    0$, approximates the partition function of the hard core model
    with activity $\lambda > 0$ on a graph $G$ drawn from $\G(n, d/n)$
    within a multiplicative factor of $(1 \pm \mu)$, and which runs in
    time polynomial in $n$ and $1/\mu$ with probability at least $1 -
    n^{-\beta}$ (over the random choice of the
    graph).
  \end{itemize}
\end{theorem}
\addtocounter{theorem}{-1}
\endgroup

}
\begin{proof}[Proof]
  Both parts of the Theorem follow easily from Theorem~\ref{thm:main-1} once we
  prove that graphs drawn from ${\cal G}(n, d/n)$ have connective
  constant at most $d(1+\epsilon/2)$ with probability at least $1 - n^{-\beta}$. 

  Recall that $N(v,\ell)$ is the number of self-avoiding walks of
  length $\ell$ starting at $v$.  Suppose $\ell \geq a\log n$, where $a$ is a constant
  depending upon the parameters $\epsilon$, $\beta$ and $d$ which will be specified later.  We first observe that
  \begin{displaymath}
    \E{\sum_{i=1}^\ell N(v,i)} \leq \sum_{i=1}^\ell\inp{\frac{d}{n}}^i n^i
    \leq d^{\ell}\frac{d}{d-1},
  \end{displaymath}
  and hence by Markov's inequality, we have $\sum_{i=1}^\ell N(v,i) \leq
  d^\ell\frac{d}{d-1}(1+\epsilon/2)^\ell$ with probability at least $1 -
  (1+\epsilon/2)^{-\ell}$.  By choosing $a$ such that
  $a\log(1+\epsilon/2) \geq \beta + 2$, we see that this probability
  is at least $1 - n^{-\inp{\beta+2}}$.  By taking a union bound over all
  $\ell$ with $a\log n \leq \ell \leq n$ and over all vertices $v$, we
  see that the connective constant $\Delta$ is at most
  $d(1+\epsilon/2)$ with probability at least $1-n^{-\beta}$.  Since
  $\lambda \leq \frac{e}{d(1+\epsilon)} < \frac{e}{d(1+\epsilon/2)}$,
  we therefore see that with probability at least $1-n^{-\beta}$, the
  conditions of the first part of Theorem~\ref{thm:main-1} are
  satisfied, and the graph sampled from ${\cal G}(n, d/n)$ exhibits
  strong spatial mixing.  This proves the part about strong spatial mixing.
  
  The algorithmic part is proved using the same algorithm as in the
  proofs of Theorems~\ref{thm:main-1} and
  \ref{thm:main-2}.  Provided the above bound on the connective
  constant holds, that algorithm will terminate
  in time $n^\gamma\poly{1/\mu}$ and produce an estimate of the partition function
  accurate up to a factor of $1 \pm \mu$, where $\gamma$ is a
  constant dependent only upon $\beta$, $\lambda$ and
  $\epsilon$.  Since the required connective constant bound holds with
  probability at least $1-n^{-\beta}$, we therefore see that the algorithm terminates in time $n^\gamma\poly{1/\mu}$ with
  at least this probability. %
\end{proof}

\section{More results on trees}
\label{sec:more-results-trees}
In this section, we improve upon the bounds in
Theorem~\ref{thm:main-1} in the special case of spherically symmetric
trees and then present an application of this improvement to the
problem of counting independent sets in bounded degree bipartite
graphs.  We also show that in the case of general trees,
uniqueness of the Gibbs measure for the hard core model can be related to
the \emph{branching factor}: a slightly stronger notion than the
connective constant that
has been used before in the context of the zero field Ising, Potts and
Heisenberg models on trees~\cite{lyons_ising_1989,pemantle_robust_1999}.

\subsection{Spherically symmetric trees}
\label{sec:spher-symm-trees}
We first consider improvements on the bounds of
Theorem~\ref{thm:main-1} in the special case of spherically symmetric
trees.  Recall that a rooted tree is called \emph{spherically
  symmetric} if the degree of any vertex in the tree depends only upon
its distance from the root.

We now consider an infinite spherically symmetric subtree rooted at
$\rho$.  Let $d_i$ denote the arity of vertices at distance $i$ from
the root.  We then have $N(\rho, \ell) = \prod_{i=0}^{\ell-1}d_i$.  The
\emph{connective constant with respect to $\rho$}, denoted by
$\Delta_\rho$, is defined as
$\Delta_\rho=\limsup_{\ell\rightarrow\infty}N(\rho,\ell)^{1/\ell}$.  Notice
that this number is always \emph{at most} the true connective
constant, which is the supremum of this quantity over all
vertices in the tree.

\begin{theorem}\label{thm:mixing-spherically-semmetric}
Let $T$ be a locally finite spherically symmetric tree rooted at
$\rho$, whose connective constant with respect to $\rho$ is
$\Delta_\rho$. If $\lambda<\lambda_c(\Delta_\rho)$, then for any two boundary conditions $\sigma$ and $\tau$ on $T$ which differ only at depth at least $\ell$ in $T$, we have
\[
\abs{R_\rho(\sigma, T) - R_\rho(\tau, T)} = \exp(-\Omega(\ell)).
\]
\end{theorem}
\begin{remark}
Note that Theorem~\ref{thm:mixing-spherically-semmetric} implies that on a spherically
symmetric tree with connective constant $\Delta$ ``as observed from the
root'', the correlation between the state of the root and a set of
vertices at a distance $\ell$ from the root decays exponentially
whenever $\lambda < \lambda_c(\Delta)$, and this rate of decay holds
irrespective of any fixed boundary conditions.  This
bound is tight as a function of $\Delta$, and improves upon the bound $\lambda <
\lambda_c(\Delta+1)$ obtained from a direct application of
Theorem~\ref{thm:main-1}.  Note, however, that this does not show
strong spatial mixing as we have defined it, since the decay of
correlations is shown only for the root.  Nevertheless, as we show in
Corollary~\ref{thm:bipartite-graphs} below,
this version of strong spatial mixing can still be applied to derive an FPTAS for the estimation
of the partition function of the hard core model on bounded degree bipartite graphs
for a range of activities larger than that promised by Weitz's
result~\cite{Weitz06CountUptoThreshold}, or our
Theorems~\ref{thm:main-1} and \ref{thm:main-2}. 
\end{remark}

We begin by proving a convexity condition on the function $\nu$
defined in Lemma~\ref{lem:tau-props} with respect to the message
$\phi$ defined in eq. (\ref{eq:4}); this condition will be a crucial
ingredient in the proofs in this section. Throughout, we
consider the hard core model with some fixed vertex activity $\lambda>
0$.  Since the potential function $\phi$ and the vertex activity
$\lambda$ are always going to be clear from the context, we will drop
the corresponding suffixes for ease of notation.
\begin{lemma}\label{lem:concavity}
  The function $H(x) \defeq\log(e^x\nu(e^x))$ is
  concave in the interval $[0,\infty)$.
\end{lemma}
\noindent The proof of this lemma is somewhat technical and can be
found in Appendix~\ref{sec:proof-lemma-refl-1}.

To take advantage of Lemma~\ref{lem:concavity}, we define the function $\chi(d) \defeq d\cdot\nu(d)$.
Lemma~\ref{lem:concavity} along with Jensen's inequality then implies
that, for any $d_i \geq 1$ and non-negative constants $\beta_i$
summing up to $1$, we have
\begin{equation}
  \prod_{i=0}^{\ell-1}\chi\inp{d_i}^{\beta_i} \leq 
  \chi\inp{\prod_{i=0}^{\ell-1}d_i^{\beta_i}}.\label{eq:20}
\end{equation}
We now prove the following analog of Lemma~\ref{lem:general-tree}
for the case of spherically symmetric trees. %
In addition to the notation used in
Lemma~\ref{lem:general-tree}, we will also need the following: for $j
\leq \ell$ we define $\Delta_{j,\ell} \defeq (\prod_{i=j}^{\ell-1}
d_i)^{1/(\ell-j)}$ (with the convention that $\Delta_{\ell, \ell} =
1$).  
Intuitively, $\Delta_{j,\ell}$ gives an estimate of the average arity
of the subtree of depth $\ell-j$ rooted at a vertex at depth $j$.
Thus, for example, $N(\rho, \ell) = \Delta_{0,\ell}^\ell$.
Further, using eq.~(\ref{eq:20}), we see that for $j \leq \ell - 1$,
\begin{equation}
  \label{eq:21}
  \chi\inp{d_j}\chi\inp{\Delta_{j+1,\ell}}^{\ell-j-1} \leq \chi\inp{\Delta_{j,\ell}}^{\ell-j}.
\end{equation}
\begin{lemma}\label{lem:symm-tree}
  Let $T$ be a locally finite spherically symmetric tree rooted at
  $\rho$. For $\ell \geq 1$, let $C_\ell$ be the cutset consisting of
  the vertices at distance exactly $\ell$ from $\rho$. Let $C' =
  C_{\ell +1}$ be the cutset comprising the children of vertices in
  $C_\ell$.  Consider two arbitrary boundary conditions $\sigma$ and
  $\tau$ on $T_{\leq C'}$ which differ only on $C'$.  With $\chi$ and
  $\Delta_{j,\ell}$ as defined above, we have
  \begin{equation*}
    |R_\rho(\sigma) - R_\rho(\tau)|^2 \leq c_2\cdot\chi\inp{\Delta_{0,\ell}}^\ell,
  \end{equation*}
  where $c_2$ is a constant depending only upon $\lambda$ (and the
  message $\phi$).
\end{lemma}
\begin{proof}
  The proof is similar in structure to the proof of
  Lemma~\ref{lem:general-tree}. The main difference is the application
  of the more delicate concavity condition provided by
  Lemma~\ref{lem:concavity} and eq.~(\ref{eq:21}) to aggregate the
  decay obtained at the children of a vertex.

  As before, for a vertex $v$ in
  $T_{\leq C'}$, we will denote by $T_v$ the subtree rooted at $v$ and
  containing all the descendants of $v$, and by $R_v(\sigma)$
  (respectively, $R_v(\tau)$) the occupation probability $R_v(\sigma,
  T_v)$ (respectively, $R_v(\tau, T_v))$ of the vertex $v$ in the
  subtree $T_v$ under the boundary condition $\sigma$ (respectively,
  $\tau)$ restricted to $T_v$.  Further, we will denote by $C_v$
  (respectively $C'_v$) the restriction of the cutset $C$
  (respectively, $C'$) to $T_v$.  We consider again the quantities
  \begin{displaymath}
    L \defeq \min_{x \in \insq{0,\lambda}}\phi'(x) = \frac{1}{\sqrt{\lambda(1+\lambda)}}\text{, and } M
    \defeq \phi(\lambda) - \phi(0) = \sinh^{-1}\inp{\sqrt{\lambda}}, 
  \end{displaymath}
  where the latter bounds are based on our specific message $\phi$
  defined in eq.~(\ref{eq:4}).

  As in the proof of Lemma~\ref{lem:general-tree}, we proceed by
  induction on the structure of $T_\rho$.  We will show that for any
  vertex $v$ in $T_\rho$ which is at a distance $j \leq \ell$ from
  $\rho$, and thus has arity $d_j$, we have
  \begin{equation}\label{eq:22}
    |\phi(R_v(\sigma)) - \phi(R_v(\tau))|^2 \leq c_1\chi(\Delta_{j,\ell})^{\ell-j}.
  \end{equation}
  where $c_1 = M^2$. To get the claim of the lemma, we notice that
  since $\rho \not\in C$, the form of the recurrence for the hard core
  model implies that both $R_\rho(\sigma)$ an $R_\rho(\tau)$ are in
  the interval $[0,\lambda]$.  It then follows that $|R_v(\sigma) -
  R_v(\tau)| \leq \frac{1}{L}|\phi(R_v(\sigma)) - \phi(R_v(\tau))|$.
  Hence, taking $v = \rho$ in eq.~(\ref{eq:22}) and then setting $c_2
  = c_1/L^2$, the claim of the lemma follows.

  We now proceed to prove eq.~(\ref{eq:22}). The base case of the
  induction consists of vertices $v$ which are either fixed by a
  boundary condition or which are in $C_\ell$.  In the first case,
  since the vertex is not in $C_\ell$, we have have $R_v(\sigma) =
  R_v(\tau)$ (since $\sigma$ and $\tau$ differ only on $C_{\ell +1}$)
  and hence the claim is trivially true.  In case $v \in C_{\ell}$,
  all the children of $v$ must lie in $C_{\ell + 1}$. The form of the
  recurrence of the hard core model then implies that both
  $R_v(\sigma)$ and $R_v(\tau)$ lie in the interval $[0,\lambda]$, so
  that we have
  \begin{displaymath}
    \abs{\phi(R_v(\sigma)) - \phi(R_v(\tau))}^2 \leq (\phi(\lambda) -
    \phi(0))^2 = M^2\text{, as required}.
  \end{displaymath}

  We now proceed to the inductive case. Consider a vertex $v$ at a
  distance $j \leq \ell - 1$ from $\rho$. Let $v_1, v_2, \ldots
  v_{d_j}$ be the children of $v$, which satisfy eq.~(\ref{eq:22}) by
  induction.  Applying Lemma~\ref{lem:tech} (and noticing that
  $\nu(d) = \xi(d)$ and $d\cdot\nu(d) = \chi(d)$) followed by the
  induction hypothesis, we then have,
  \begin{align*}
    \abs{\phi(R_v(\sigma)) - \phi(R_v(\tau))}^2 &\leq
    \nu(d_j)\sum_{i=1}^{d_j}\abs{\phi(R_{v_i}(\sigma)) -
      \phi(R_{v_i}(\tau))}^2\text{, using Lemma~\ref{lem:tech}}\\
    &\leq c_1\chi(d_j)\chi(\Delta_{j+1,\ell})^{\ell-j-1}\text{, using
      the induction hypothesis}\\
    &\leq c_1\chi(\Delta_{j,\ell})^{\ell-j}\text{, using eq.~(\ref{eq:21}).}
  \end{align*}
  This completes the induction.
\end{proof}

\begin{proof}[Proof of Theorem~\ref{thm:mixing-spherically-semmetric}]
Recall that $\chi(d)=d\cdot \nu(d)$. We then see from
Lemma~\ref{lem:tau-props} that $\chi(d)$ is increasing in $d$ and
$\chi(\Delta_\rho)<1$ for $\lambda<\lambda_c(\Delta_\rho)$. Note that
$\Delta_{0,\ell}\le\Delta_\rho$ for $\ell$ large enough, and hence the claim in the theorem follows form Lemma~\ref{lem:symm-tree}.
\end{proof}
\begin{corollary}\label{thm:bipartite-graphs}
  Consider the family of bipartite graphs in which the ``left'' side
  has degree at most $(d_1+1)$ and the ``right'' side has degree at
  most $(d_2 + 1)$.  Let $0 < \lambda < \lambda_c(\sqrt{d_1d_2})$.
  Then,
  \begin{itemize}
  \item The hard core model with vertex activity $\lambda$ exhibits
    strong spatial mixing on this family of graphs.
  \item There is an FPTAS for the partition function of the hard core
    model with vertex activity $\lambda$ for graphs in this family.
  \end{itemize}
\end{corollary}
Notice that for this class of graphs, the above corollary improves
upon Weitz's result~\cite{Weitz06CountUptoThreshold} (which is valid
only for $\lambda < \lambda_c(\max (d_1, d_2))$), as well as
Theorem~\ref{thm:main-2}: the latter would require $\nu(\max (d_1,
d_2))\sqrt{d_1d_2} \leq 1$, while the requirement in the above
corollary is equivalent to $\nu\inp{\sqrt{d_1d_2}}\sqrt{d_1d_2} \leq
1$.
\begin{proof}
  The proof is similar in structure to the proofs of
  Theorem~\ref{thm:main-1} and \ref{thm:main-2}, except that the
  tighter Lemma~\ref{lem:symm-tree} is used to analyze the decay of
  correlations in place of Lemma~\ref{lem:general-tree} used in those
  proofs.  
  
  We first prove that the hard core model with these parameters
  exhibits strong spatial mixing on this family of graphs. Let $G$ be
  any graph from $\family{F}$, and let $v$ be any vertex in $G$, and consider
  any boundary conditions $\sigma$ and $\tau$ in $G$ which differ only
  at a distance of at least $\ell + 1$ from $v$.  Consider the Weitz
  self-avoiding walk tree~\cite{Weitz06CountUptoThreshold} $\Ts{v, G}$
  rooted at $v$.  Without loss of generality, assume that $v$ lies on
  the ``left'' side of the graph.

  We observe that $\Ts{v,G}$ is a subtree of a spherically symmetric
  tree ${\cal T}$ in which the root has degree $d_1 + 1$, other
  vertices at even distance from the root have arity $d_1$, and
  vertices at odd distance from the root have arity $d_2$.  Notice that
  we may, in fact, view $\Ts{v, G}$ as just the tree ${\cal T}$ by
  adding in the extra boundary condition that the vertices of ${\cal
    T}$ that are not present in $\Ts{v, G}$ are fixed to be
  ``unoccupied''.  This latter modification does not modify any of the
  occupation probabilities.  With a slight abuse of notation, we refer
  to ${\cal T}$ with these boundary conditions as $\Ts{v, G}$.
  Further, as before, we denote the extra boundary conditions on
  $\Ts{v, G}$ corresponding to the boundary condition $\sigma$
  (respectively, $\tau$) on $G$ by the same letter $\sigma$
  (respectively, $\tau$).  From Weitz's theorem, we then known that
  $R_v(\sigma, G) = R_v(\sigma, \Ts{v, G})$ (respectively,
  $R_v(\sigma, G) = R_v(\sigma, \Ts{v, G})$).

  Now, let $C_\ell$ denote the cutset in $\Ts{v, G}$ consisting of all
  vertices at distance $\ell$ from $v$, for some $\ell \geq
  1$. Applying Lemma~\ref{lem:symm-tree} and adjusting for effects
  due to the parity of $\ell$ as well as the possibly higher arity
  (at most $\max\inp{d_1 + 1, d_2 + 1}$) at the root, we then have 
  \begin{equation}
    \abs{R_v(\sigma, G) - R_v(\tau, G)}^2 = \abs{R_v(\sigma, \Ts{v,G})
      - R_v(\tau, \Ts{v,G})}^2 \leq c_2M_0^2\chi(\Delta)^\ell\label{eq:23}
  \end{equation}
  where $M_0 = \max\inp{\chi(d_1+1), \chi(d_2+1)}/\chi(\Delta)$ and
  $\Delta = \sqrt{d_1d_2}$ are constants, and $c_2$ is the constant in
  the statement of Lemma~\ref{lem:symm-tree}.  Recalling that
  $\chi(\Delta) = \Delta\nu(\Delta)$, we see from
  Lemma~\ref{lem:tau-props} that $\chi(\Delta) < 1$ for $\lambda <
  \lambda_c(\Delta)$.  Combining with eq.~(\ref{eq:23}), this
  establishes SSM in the regime $\lambda < \lambda_c(\sqrt{d_1d_2})$
  and hence proves the first part of the corollary.

  The proof of the algorithmic part of the corollary is virtually
  identical to the proof of the algorithmic part of
  Theorem\nobreakspace \ref {thm:main-2}. As before, the depth $\ell$
  up to which the self-avoiding walk tree needs to be explored in
  order to additively approximate $R_v$ up to an error of
  at most $\mu$ is $ O(1) + O(\log (1/\mu)) = O(\log n + \log
  (1/\mu))$, and hence the running time $\Delta^{O(\ell)}$ is still
  polynomial in $n$ and $1/\mu$.    
\end{proof}

\subsection{Branching factor and uniqueness of the Gibbs measure on
  general trees}
\label{sec:branch-numb-uniq}
We close with an application of our results to finding thresholds
for the uniqueness of the Gibbs measure of the hard core model on locally
finite infinite trees.  Our bounds will be stated in terms of the
\emph{branching factor}, which, as indicated above, has been shown to
be the appropriate parameter for establishing phase transition
thresholds for symmetric models such as the ferromagnetic Ising, Potts
and Heisenberg models~\cite{lyons_ising_1989,pemantle_robust_1999}.
We begin with a general definition of the notion of uniqueness of
Gibbs measure (see, for example, the survey article of
Mossel~\cite{mossel_survey:_2004}).  Let $T$ be a locally finite
infinite tree rooted at $\rho$, and let $C$ be a cutset in $T$.
Consider the hard core model with vertex activity $\lambda >0$ on $T$.
We define the \emph{discrepancy} $\delta(C)$ of $C$ as follows.  Let
$\sigma$ and $\tau$ be boundary conditions in $T$ which fix the state
of the vertices on $C$, but not of any vertex in $T_{<C}$.  Then,
$\delta(C)$ is the maximum over all such $\sigma$ and $\tau$ of the
quantity $R_\rho(\sigma, T) - R_\rho(\tau,T)$.
\begin{definition}
  \textbf{(Uniqueness of Gibbs measure).} The hard core model with
  vertex activity $\lambda > 0$ is said to exhibit \emph{uniqueness of
    Gibbs measure} on $T$ if there exists a sequence of cutsets
  $\inp{B_i}_{i=1}^\infty$ such that $\lim_{i\rightarrow\infty}
  d(\rho, B_i) \rightarrow \infty$ and such that
  $\lim_{i\rightarrow\infty}\delta(B_i) = 0$. 
\end{definition}
\begin{remark}
  Our definition of \emph{uniqueness} here is similar in form to those
  used by Lyons~\cite{lyons_ising_1989} and Pemantle and Steif~\cite{pemantle_robust_1999}.  Notice, however, that the recurrence
  for the hard core model implies that the discrepancy is
  ``monotonic'' in the sense that if cutsets $C$ and $D$ are such that
  $C < D$ (i.e., every vertex in $D$ is the descendant of some vertex
  in $C$) then $\delta(C) > \delta(D)$.  This ensures that the choice
  of the sequence $\inp{B_i}_{i=1}^\infty$ in the definition above is
  immaterial.  For example, uniqueness is defined by
  Mossel~\cite{mossel_survey:_2004} in terms of the cutsets $C_\ell$
  consisting of vertices at distance exactly $\ell$ from the root.
  However, the above observation shows that for the hard core model,
  Mossel's definition is equivalent to the one presented here.
\end{remark}

We now define the notion of the branching factor of an infinite tree.
\begin{definition}
  \textbf{(Branching
    factor~\cite{lyons_ising_1989,lyons_random_1990,pemantle_robust_1999}).}
  Let $T$ be an infinite tree.  The branching factor $\br{T}$ is
  defined as follows:
  \begin{equation*}
    \br{T} \defeq \inf\inb{b > 0\left| \inf_C\sum_{v\in C}b^{-|v|} = 0
      \right.},
  \end{equation*}
  where the second infimum is taken over all cutsets $C$.
\end{definition}
To clarify this definition, we consider some examples.  If $T$ is a
$d$-ary tree, then $\br{T} = d$.  Further, by taking the second
infimum over the cutsets $C_\ell$ of vertices at distance $\ell$ from
the root, it is easy to see that the branching factor is never more
than the connective constant.  Further, Lyons~\cite{lyons_random_1990}
observes that in the case of spherically symmetric trees, one can define
the branching factor as
$\liminf_{\ell\rightarrow\infty}N(\rho,\ell)^{\ell}$.

We are now ready to state and prove our results on the uniqueness of
the hard-core model on general trees. 
\begin{theorem}\label{thm:uniq-tree}
  Let $T$ be an infinite tree rooted at $\rho$ with branching factor
  $b$.  The hard core model with vertex activity $\lambda > 0$
  exhibits uniqueness of Gibbs measure on $T$ if at least one of the
  following conditions is satisfied:
 \begin{itemize}
 \item $\lambda < \lambda_c(b + 1)$.
 \item $T$ is a spherically symmetric tree and $\lambda <
   \lambda_c(b)$.
 \end{itemize}
\end{theorem}
Notice that the result for spherically symmetric trees is tight.  We conjecture, however, that the general case can also
be improved to $\lambda_c(b)$.  

\begin{proof}[Proof of Theorem~\ref{thm:uniq-tree}]
  We first consider the case of general trees. We will apply
  Lemma~\ref{lem:general-tree} specialized to the message
  in eq.~(\ref{eq:4}).  From the proof of item \ref{item:2} of
  Theorem~\ref{thm:main-1}, we recall that when $\lambda < \lambda_c(b
  +1)$, $\alpha \defeq \sup_{d \geq 1}\nu(d) < \frac{1}{b}$, while
  $\alpha' \defeq \inf_{d\geq 1}\nu(d) = \nu(1)$ is a positive constant.  Suppose
  $\alpha = \frac{1}{b(1+\epsilon)}$ for some $\epsilon > 0$.
  Applying Lemma~\ref{lem:general-tree} to an arbitrary cutset $C$, we
  then get
  \begin{equation}
    \delta(C)^2 \leq M_0\sum_{v \in C}\insq{(1+\epsilon)b}^{-|v|},\label{eq:25}
  \end{equation}
  where $M_0$ is a constant.  Since $b(1+\epsilon) > \br{T}$, the definition of $\br{T}$ implies
  that we can find a sequence $\inp{B_i}_{i=1}^\infty$ of cutsets such
  that
  \begin{equation*}
    \lim_{i\rightarrow\infty}\sum_{v \in
      B_i}\insq{(1+\epsilon)b}^{-|v|} = 0.%
  \end{equation*}
  Further, such a sequence must satisfy
  $\lim_{i\rightarrow\infty}d(\rho, B_i) = \infty$, since otherwise
  the limit above would be positive.  Combining with
  eq.~(\ref{eq:25}), this shows that
  $\lim_{i\rightarrow\infty}\delta(B_i) = 0$, which completes the
  proof of this case.

  For the case of the spherically symmetric tree, we will use Lyons's
  observation~\cite{lyons_random_1990} that for such trees
  \begin{equation}\label{eq:26}
    b = \br{T} = \liminf_{n\rightarrow\infty}N(v,\ell)^{1/\ell}. 
  \end{equation}
  Again, from the fact that $\lambda < \lambda_c(b)$, and the
  properties of the function $\nu$ proved in
  Lemma~\ref{lem:tau-props}, we see that there exists an $\epsilon >
  0$ such that for $b_1 < b(1+\epsilon)$,
  \begin{equation}\label{eq:29}
    \chi(b_1) = b_1\cdot\nu(b_1) \leq 1 - \epsilon.
  \end{equation}
  Further, eq.~(\ref{eq:26}) implies that there exists a
  strictly increasing sequence $\inp{\ell_i}_{i=1}^\infty$, such that the
  cutsets $C_{\ell_i}$ (of vertices at distance exactly $\ell_i$ from the
  root) satisfy
  \begin{equation}\label{eq:28}
    \abs{C_{\ell_i}}^{1/\ell_i} =
    N(\rho,\ell_i)^{1/\ell_i} =
    \Delta_{0, \ell_i} < b(1+\epsilon),
  \end{equation}
  where $\Delta_{0,\ell} = N(\rho,\ell)^{1/\ell}$ is as defined above.
  Applying Lemma~\ref{lem:symm-tree} to a cutset $C_{\ell_i}$, we then
  have 
  \begin{equation*}
    \delta(C_{\ell_i})^2 \leq M_1\chi(\Delta_{0,\ell_i})^{\ell_i},
  \end{equation*}
  where $M_1$ is a constant.  Combining with eqs.~(\ref{eq:29}) and
  (\ref{eq:28}), we then see that
  $\lim_{i\rightarrow\infty}\delta(C_{\ell_i}) = 0$, which completes the
proof.
\end{proof}

  \paragraph{Acknowledgments.} We thank Elchanan Mossel, Allan Sly and
  Dror Weitz for helpful discussions.%

\appendix
  
\section{Monotonicity properties of $\nu_\lambda$}
\label{sec:monot-prop-nu_l}
In this section, we prove Lemma~\ref{lem:tau-props}. 
\begin{proof}[Proof of Lemma~\ref{lem:tau-props}]
  The expression for $\nu_\lambda(d)$ follows by using
  the definitions $\nu_\lambda (d) = \xi(d)=
  \Xi(d, \tilde{x}_\lambda(d))$ and then using eq.\nobreakspace \textup {(\ref {eq:6})}
  from Lemma\nobreakspace \ref {lem:fixed-point}.  To prove the rest of item\nobreakspace \ref {item:1}, we
  observe that when $\lambda = \lambda_c(d) =
  \frac{d^d}{(d-1)^{d+1}}$, $\tilde{x}_{\lambda}(d) = \frac{1}{d-1}$
  is the unique solution to eq.\nobreakspace \textup {(\ref {eq:6})} (indeed, the potential
  function $\phi$ was chosen in \cite{li_correlation_2011} to achieve
  this condition).  Plugging this into the definition of
  $\nu_\lambda$, we see that $\nu_{\lambda}(d) = \frac{1}{d}$, as
  claimed.
  
  We now proceed to prove the other two items.  For item\nobreakspace \ref {item:2}, we
  compute the derivative $\nu_\lambda'(d)$ using the chain rule, and
  show that it is positive (for ease of notation, we denote
  $\tilde{x}_\lambda(d)$ as $\tilde{x}$):
  \begin{align}
    \nu_\lambda'(d) &= \Xi^{(1,0)}(d,\tilde{x}) +
    \Xi^{(0,1)}(d,\tilde{x})\diff{\tilde{x}}{d}\nonumber\\
    &=\Xi^{(1,0)}(d,\tilde{x}), \text{ since $\Xi^{(0,1)}(d,\tilde{x})
      = 0$ by definition of $\tilde{x}$}\nonumber\\
    &=\Xi(d,\tilde{x})\insq{\frac{1}{d} -
      \frac{\log(1+\tilde{x})}{1+f_{d,\lambda}(\tilde{x})}}\nonumber\\
    &=\frac{\Xi(d, \tilde{x})}{d\tilde{x}}\insq{\tilde{x} -
      \log(1+\tilde{x})} > 0, \text{ since $\tilde{x} > 0$ for $d >
      0$.}\label{eq:18}
  \end{align}
  Here, we use $1 + f_{d,\lambda}(\tilde{x}) = d\tilde{x}$ to get the
  last equality. 
  
  To prove item\nobreakspace \ref {item:3}, we notice first that for fixed $d$,
  $\nu_\lambda(d) = \frac{d\tilde{x}_\lambda(d) - 1}{1+
    \tilde{x}_\lambda(d)}$ increases with $\tilde{x}_\lambda(d)$.
  Thus, we only need to establish that $\tilde{x}_\lambda(d)$ is
  increasing in $\lambda$ when $d$ is fixed.  We first observe that
  eq.\nobreakspace \textup {(\ref {eq:6})} implies that
  \begin{equation}
    \lambda = (d\tilde{x}_{\lambda}(d) - 1)(1
    + \tilde{x}_\lambda(d)).\label{eq:8}
\end{equation}
Since $\lambda > 0$, we must have $d\tilde{x}_{\lambda}(d) - 1 \geq
0$, in which case the right hand side of eq.\nobreakspace \textup {(\ref {eq:8})} increases with
$\tilde{x}_\lambda(d)$.  This shows that as $\lambda$ increases, so
must $\tilde{x}_\lambda(d)$. This completes the proof.
\end{proof}

\section{Description of numerical results}
\label{sec:descr-numer-results}
In this section, we describe how the numerical bounds presented in Table~\ref{fig:1}
were obtained.  All of the bounds are direct
applications of Theorem~\ref{thm:main-2} using published upper bounds
on the connective constant for the appropriate graph.  For the Cartesian
lattices $\Z^2, \Z^3, \Z^4, \Z^5$ and $\Z^6$, and the triangular
lattice~$\mathbb{T}$, the exact connective constant is not known, but
rigorous upper and lower bounds are available in the
literature~\cite{madras96:_self_avoid_walk,weisstein:_self_avoid_walk_connec_const}.
In a recent breakthrough, Duminil-Copin and
Smirnov~\cite{duminil-copin_connective_2012} rigorously established
that connective constant of the honeycomb lattice $\mathbb{H}$ is
$\sqrt{2 + \sqrt{2}}$.  In order to apply Theorem~\ref{thm:main-2} for
a given lattice of maximum degree $d+1$ and connective constant
$\Delta$, we choose a given $\lambda$, solve eq.~(\ref{eq:6})
for $\tilde{x}_\lambda(d)$ (this is a polynomial equation of degree
$d+1$, and hence can be solved efficiently for all these lattices
which have small degree), and then compute $\nu_\lambda(d)\Delta$ to
check if it is less than $1$.  The monotonicity of $\nu_\lambda(d)$
in $\lambda$ (Lemma~\ref{lem:tau-props}) then allows us to search
easily for the best possible $\lambda$.  Each of these computations
for a particular lattice takes a few seconds on
\emph{Mathematica}~\cite{wolfram09:_mathem} running on a laptop with a
2.8 GhZ Intel\textsuperscript{\textregistered}
Core\textsuperscript{\texttrademark} 2 Duo CPU, and 4GB of RAM.

As we pointed out in the introduction, any improvement in the
connective constant of a lattice (or that of its corresponding Weitz SAW tree) will immediately lead to an improvement in our bounds.  We
demonstrate this here in the special case of $\Z^2$, by using a tighter
combinatorial analysis of the connective constant for the
Weitz SAW tree of this lattice to obtain a somewhat better
bound than the one presented in Table~\ref{fig:1}.  The Weitz SAW tree adds
additional boundary conditions to the SAW tree of the lattice, and
hence allows a smaller number of self-avoiding walks, and therefore
can have a smaller connective constant than that of the lattice
itself.  Further, the proof of Theorem~\ref{thm:main-2} is only in
terms of the SAW tree, and hence the bounds there clearly hold if the
connective constant of the SAW tree is used in place of the connective
constant of the lattice.  

To upper bound the connective constant of the Weitz SAW tree, we use
the method of \emph{finite memory self-avoiding
  walks}~\cite{madras96:_self_avoid_walk}---these are walks which are
constrained only to not have cycles of length up to some finite length
$L$. Clearly, the number of such walks of any given length $\ell$ upper
bounds $N(v, \ell)$.  In order to bring the boundary conditions on the
Weitz SAW tree into play, we
further enforce the constraint that the walk is not allowed to make
any moves which will land it in a vertex fixed to be ``unoccupied'' by
Weitz's boundary conditions.  Such a walk can be in one of a finite
number $k$ (depending upon $L$) of states, such that the number of
possible moves it can make to state $j$ while respecting the above
constraints is some finite number $M_{ij}$.  The $k\times k$ matrix $M
= (M_{ij})_{i,j\in[k]}$ is
called the \emph{branching
  matrix}~\cite{restrepo11:_improv_mixin_condit_grid_count}. We
therefore get $N(v,\ell) \leq \vec{e_1}^TM^\ell\vec{1}$, where $\vec{1}$
denotes the all $1$'s vector, and $\vec{e_1}$ denotes the co-ordinate
vector for the state of the zero-length walk.

Since the entries of $M$ are non-negative, the Perron-Frobenius
theorem implies that one of the maximum magnitude eigenvalues of
the matrix $M$ is a positive real number $\gamma$.  Using the Jordan
canonical form, one then sees that
\begin{displaymath}
  \limsup_{\ell\rightarrow\infty} N(v,\ell)^{1/\ell} \leq 
  \limsup_{\ell\rightarrow\infty}(\vec{e_1}^TM^\ell\vec{1})^{1/\ell}
  \leq \max\inp{\gamma,1}.
\end{displaymath}
Hence, the largest real eigenvalue $\gamma$ of $M$ gives a bound on
the connective constant of the Weitz SAW tree.  Using the matrix $M$
corresponding to walks in which cycles of length at most $L=14$ are
avoided, we get that the connective constant of the Weitz SAW tree is
at most $2.5384$.  Using this bound, and applying
Theorem~\ref{thm:main-2} as described above, we get the bounds $2.614$
and $2.185$ for $\Delta$ and $\lambda$ respectively, in the notation
of the table.

\section{Proofs omitted from Section~\ref{sec:messages-tree}}
\label{sec:proof-lemma-refl}

\begin{proof}[Proof of Lemma~\ref{lem:mean-value}]
  Define $F(t) = f_{d,\lambda}^\phi(t\vec{x} + (1-t)\vec{y})$ for $t \in
  [0,1]$.  By the scalar mean value theorem applied to $F$, we have 
  \begin{equation*}
    f_{d,\lambda}^\phi(\vec{x}) - f_{d,\lambda}^\phi(\vec{y}) = F(1) -
    F(0) = F'(s)\text{, for some $s \in [0,1]$}.
  \end{equation*}
  Let $\psi$ denote the inverse of the message $\phi$: the derivative
  of $\psi$ is given by $\psi'(y) = \frac{1}{\Phi(\psi(y))}$,
  where $\Phi$ is the derivative of $\phi$.  We now define the vector
  $\vec{z}$ by setting $z_i = \psi(sx_i + (1-s)y_i)$ for $1 \leq i
  \leq d$.  We then have
  \begin{align}
    \abs{f_{d,\lambda}^\phi(\vec{x}) - f_{d,\lambda}^\phi(\vec{y})} &=
    \abs{F'(s)} = \abs{\ina{\nabla{f_{\lambda, d}^\phi(s\vec{x} +
        (1-s)\vec{y})},
        \vec{x} - \vec{y}}}\nonumber\\
    &= \Phi(f_{d,\lambda}(\vec{z}))\abs{\sum_{i=1}^d\frac{x_i -
        y_i}{\Phi(z_i)}\pdiff{f_{d,\lambda}}{z_i}}, \quad\text{using the chain rule}\nonumber\\
    &\leq
    \Phi\inp{f_{d,\lambda}(\vec{z})}\sum_{i=1}^d\frac{\abs{y_i-x_i}}
    {\Phi(z_i)}\abs{\pdiff{f_{d,\lambda}}{z_i}}, \quad\text{as claimed.}\nonumber
  \end{align}
  We recall that for simplicity, we are using here the somewhat
  non-standard notation $\pdiff{f}{z_i}$ for the value of the partial
  derivative $\pdiff{f}{R_i}$ at the point $\vec{R} = \vec{z}$.
\end{proof}

\begin{proof}[Proof of Lemma~\ref{lem:general-tree}]
  We first set up some notation. For a vertex $v$ in $T_{\leq C'}$, we
  will denote by $T_v$ the subtree rooted at $v$ and containing all
  the descendants of $v$.  By a slight abuse of notation, we denote by
  $R_v(\sigma)$ (respectively, $R_v(\tau)$) the occupation probability
  $R_v(\sigma, T_v)$ (respectively, $R_v(\tau, T_v))$ of the vertex
  $v$ in the subtree $T_v$ under the boundary condition $\sigma$
  (respectively, $\tau)$ restricted to $T_v$.  Further, we will denote
  by $C_v$ (respectively $C'_v$) the restriction of the cutset $C$
  (respectively, $C'$) to $T_v$.  We also define the following two
  quantities related to the message $\phi$:
  \begin{displaymath}
    L \defeq \min_{x \in \insq{0,\lambda}}\phi'(x)\text{, and } M
    \defeq \phi(\lambda) - \phi(0)
  \end{displaymath}
  Notice that both these quantities are finite and positive because of
  the constraints in the definition of a message. 

  By induction on the structure of $T_\rho$, we will now show that for
  any vertex $v$ in $T_\rho$ which is at a distance $\delta_v$ from
  $\rho$, and has arity  $d_v$, one has
  \begin{equation}\label{eq:7}
    |\phi(R_v(\sigma)) - \phi(R_v(\tau))|^q \leq \frac{c_1\xi(d_v)}{\alpha'}
      \sum_{u \in C_v} \alpha^{|u|-\delta_v},
  \end{equation}
  where $c_1 = M^q$. To get the claim of the lemma, we notice
  that since $\rho \not\in C$, the form of the recurrence for the hard
  core model implies that both $R_\rho(\sigma)$ an $R_\rho(\tau)$ are
  in the interval $[0,\lambda]$.  It then follows that $|R_v(\sigma) -
  R_v(\tau)| \leq \frac{1}{L}|\phi(R_v(\sigma)) - \phi(R_v(\tau))|$.
  Hence, taking $v = \rho$ in eq.~(\ref{eq:7}) and then setting $c_0 =
  c_1/L^q$, the claim of the lemma follows.

  We now proceed to prove eq.~(\ref{eq:7}). The base case of the
  induction consists of vertices $v$ which are either of arity $0$ or
  which are in $C$.  In the first case (which includes the case where
  $v$ is fixed by the boundary condition), we clearly have
  $R_v(\sigma) = R_v(\tau)$, and hence the claim is trivially true.
  In the second case, we have $C_v = \inb{v}$, and all the children of
  $v$ must lie in $C'$. The form of the recurrence of the hard core
  model then implies that both $R_v(\sigma)$ and $R_v(\tau)$ lie in
  the interval $[0,\lambda]$, so that we have 
  \begin{displaymath}
    \abs{\phi(R_v(\sigma)) - \phi(R_v(\tau))}^q \leq (\phi(\lambda) -
    \phi(0))^q = M^q \leq
    \frac{M^q\xi(d_v)}{\alpha'}\text{, since $\alpha' \leq \xi(d_v)$.}
  \end{displaymath}

  We now proceed to the inductive case.  Let $v_1, v_2, \ldots
  v_{d_v}$ be the children of $v$, which satisfy eq.~(\ref{eq:7}) by
  induction.  Applying Lemma~\ref{lem:tech} followed by the induction
  hypothesis, we then have,
  \begin{align*}
    \abs{\phi(R_v(\sigma)) - \phi(R_v(\tau))}^q &\leq
    \xi(d_v)\sum_{i=1}^{d_v}\abs{\phi(R_{v_i}(\sigma)) -
      \phi(R_{v_i}(\tau))}^q\text{, using Lemma~\ref{lem:tech}}\\
    &\leq
    \frac{c_1\xi\inp{d_v}}{\alpha'}\sum_{i=1}^{d_v}\xi\inp{d_{v_i}}\sum_{u
      \in C_{v_i}}\alpha^{|u| - \delta_{v_i}}\text{, using the
      induction hypothesis}\\
    &\leq\frac{c_1\xi\inp{d_v}}{\alpha'}\sum_{u \in
      C_v}\alpha^{|u|-\delta_v}\text{, using $\xi(d_{v_i}) \leq
      \alpha$ and $\delta_{v_i} = \delta_{v} + 1$}.
  \end{align*}
  This completes the induction.  
\end{proof}

\section{Proof of Lemma~\ref{lem:concavity}}
\label{sec:proof-lemma-refl-1}
\begin{proof}[Proof of Lemma~\ref{lem:concavity}]
  Define the function $L$ as
  \begin{displaymath}
    L(x) \defeq \frac{\nu'(x)}{\nu(x)}.
  \end{displaymath}
  Notice that by item \ref{item:2} in Lemma~\ref{lem:tau-props}, $L(x)
  \geq 0$ for $x \geq 1$. In particular, $L(e^z) \geq 0$ for $z \geq
  0$. We can now write the second derivative of $H$ as follows:
  \begin{equation*}%
    H''(x) = e^zL(e^x)\inp{1 + e^x\frac{L'(e^x)}{L(e^x)}}.
  \end{equation*}
  Since $L(e^z) \geq 0$ for non-negative $z$ as observed above, in order to show that $H$ is concave for $x \geq 0$, we
  only need to show that 
  \begin{equation}\label{eq:17}
    d\frac{L'(d)}{L(d)} \leq -1\text{ for all $d \geq 1$}
  \end{equation}
   We now proceed
  to analyze the function $L$.  We recall that by the definition of
  $\nu$ following Lemma~\ref{lem:fixed-point}, we have $\nu(d) =
  \Xi(d, \tilde{x}(d))$, where $\Xi$ is as defined in that lemma and
  $\tilde{x}(d)$ is the unique non-negative solution of
  $\Xi^{(0,1)}(d, x) = 0$.  For ease of notation, we denote
  $\tilde{x}(d)$ by $\tilde{x}$ in what follows.  We now have
  \begin{align*}
    L(d) &= \frac{\nu'(d)}{\nu(d)} =
    \frac{\nu'(d)}{\Xi(d,\tilde{x})}\\
    &=\frac{1}{d\tilde{x}}\insq{\tilde{x} - \log(1+\tilde{x})}\text{,
      using the derivation of eq.~(\ref{eq:18}) above.}
  \end{align*}
  Before proceeding further, we evaluate $\tilde{x}' \defeq
  \diff{\tilde{x}}{d}$.  Since $d\tilde{x} = 1 + f_d(\tilde{x})$,
  differentiation gives
  \begin{displaymath}
    \tilde{x} + d\tilde{x}' = 
    -f_d(\tilde{x})\insq{\frac{d\tilde{x}'}{1+\tilde{x}}  + \log(1+\tilde{x})}
  \end{displaymath}
  which in turn yields (using $d\tilde{x} = 1 + f_d(\tilde{x})$)
  \begin{align*}
    \tilde{x}' &= -\frac{(1+\tilde{x})
      \insq{f_d(\tilde{x})\log(1+\tilde{x}) +
        \tilde{x}}}{d(1+d)\tilde{x}}.%
  \end{align*}
  Since $\tilde{x} \geq 0$, this shows that $\tilde{x}' \leq 0$.  We
  now differentiate $L(d)$ to get
  \begin{align*}
    \frac{L'(d)}{L(d)} &= -\frac{1}{d} +
    \tilde{x}'\insq{
      \frac{(1+\tilde{x})\log(1+\tilde{x})-\tilde{x}}
      {\tilde{x}(1+\tilde{x})(\tilde{x}-\log(1+\tilde{x}))}}.
  \end{align*}
  Since $\tilde{x} \geq 0$, we have both
  $(\tilde{x}-\log(1+\tilde{x})) \geq 0$ and
  $(1+\tilde{x})\log(1+\tilde{x})-\tilde{x} \geq 0$.  Combining this
  with the observation above that $\tilde{x}' \leq 0$, this shows that
  $d\frac{L'(d)}{L(d)} \leq -1$ for $d \geq 1$.  As observed in the
  discussion following eq.~(\ref{eq:17}), this implies that $H$ is
  concave for $x \geq 0$.  
\end{proof}

\end{document}